\documentclass[journal]{IEEEtran}
\usepackage[latin1]{inputenc}
\usepackage{graphicx}
\usepackage{amsmath,amsthm,amssymb}
\usepackage{csquotes}
\usepackage{enumerate}
\usepackage{enumitem}
\usepackage{cite}
\usepackage{hyperref}
\usepackage{lipsum}

\newcommand{\norm}[1]{\left\lVert#1\right\rVert}
\newcommand{\HH}{\rm{H}}
\newcommand{\mf}{ \mathbf }
\newcommand{\TT}{\rm{T}}
\DeclareMathOperator{\E}{\rm{E}}
\newcommand{\quotes}[1]{``#1''}

\newtheorem{thm}{Theorem}
\newtheorem{corollary}{Corollary}
\newtheorem{lemma}{Lemma}
\newtheorem{remark}{Remark}

\begin{document}
	
\title{Uplink Achievable Rate in One-bit Quantized Massive MIMO with Superimposed Pilots}
\author{M.A.~Teeti,~\IEEEmembership{Member,~IEEE,} Rui~Wang,~\IEEEmembership{Member,~IEEE,} and Reza~Abdolee,~\IEEEmembership{Member,~IEEE,}%
	
	\thanks{M. Teeti and R. Wang are with the Department of Electrical and Electronic Engineering, Southern University of Science \& Techology, Shenzhen, 518055, China (e-mail:moha@sustc.edu.cn, wang.r@sustc.edu).}%
	\thanks{R. Abdolee is with  the Department of Computer \& Electrical and Computer Science, California State University, Bakersfield, 93311 California (email: rabdolee@csub.edu).}}

\maketitle


\begin{abstract}
 In this work, we consider a 1-bit quantized massive MIMO channel with superimposed pilot (SP) scheme, dubbed QSP. With linear minimum mean square error (LMMSE) channel estimator and maximum ratio combining (MRC) receiver at the BS, we derive an approximate lower bound on the achievable rate. When optimizing pilot and data powers, the optimal power allocation maximizing the data rate is obtained in a closed-form solution. Although there is a performance gap between the quantized and unquantized systems, it is shown that this gap diminishes as the number of BS antennas is asymptotically large. Moreover, we show that pilot removal from the received signal by using the channel estimate doesn't result in a significant increase in information, especially in the cases of low signal-to-noise ratio (SNR) and a large number of users. We present some numerical results to corroborate our analytical findings and insights are provided for further exploration of the quantized systems with SP.
\end{abstract}





\section{Introduction}
\label{sec:sec_1}
{E}{quipping} the base station (BS) with low-resolution analog-to-digital converters (ADCs) or digital-to-analog converters (DACs) is highly appealing in massive multiple-input multiple-output (MIMO) system, owing to the substantial reduction in hardware complexity, energy consumption and the amount of baseband data generated at the BS~\cite{6731024,6891254}. These three considerations are important in massive MIMO, specifically in enabling the new emerging mmWave MIMO technology~\cite{Rangan2014,Orhan2015,Heath2015}, where a much broader bandwidth than the traditional sub-6 GHz band is included and a particularly large number of antennas is employed at the BS. The 1-bit ADC, which is our interest in this work, has a simple structure which boils down to a single comparator with a relatively negligible energy consumption~\cite{Walden1999}. Also, the compensation for variations in the received signal level by automatic gain control is needless. Promoted by such simplicity, there has been an increasing interest in replacing the high-resolution ADCs at the BS by the 1-bit ADCs~\cite{Risi2014,Mo2015,Fan2015,Choi2016,Liang2016,7600443,Li2016}. According to this new architecture, each antenna at the BS is equipped with a pair of 1-bit ADCs, for both the in-phase and quadrature components of the received complex-baseband signal (see the illustration in Fig.~\ref{fig:one-bit system}(a)).  

The impact of the nonlinear distortion of quantization on channel capacity is studied in the literature under various assumptions on the channel state information at both transmitter (CSIT) or receiver (CSIR). With perfect CSIT and CSIR, the capacity of the 1-bit quantized real-valued single-input single-output (SISO) Gaussian channel is studied in~\cite{Singh2006}, where the results show that antipodal signaling is capacity achieving. For fading SISO channel, it turns out that quadrature phase-shift keying (QPSK) signaling is optimal~\cite{4594988}~\cite{Krone2010}~\cite{Mo2015}. A more general transmit-receive antenna configurations with the 1-bit ADC is considered in ~\cite{Mo2015}, where it is shown that QPSK signaling combined with the maximum ratio transmission is optimal for the multiple-input single-output channel (MISO) channel.

In all above works on the capacity of quantized channels, perfect CSI is assumed, while in practice the channel requires being estimated beforehand at the receiver by the aid of dedicated training pilots, for instance. Since quantization process throws away some information about the channel, it turns out that gaining reliable CSI is a challenging issue in quantized systems, especially when the 1-bit ADCs are used. In massive MIMO, the accuracy of CSI at the BS is one of the key requirements to harness its high spectral efficiency~\cite{5595728,6457363}. Thus understanding the performance gap between the quantized and unquantized massive MIMO is important. Since determining the exact capacity of quantized MIMO channel is hard, many researchers have focused on the achievable rate in massive MIMO and the possibility of supporting high-order modulation, while assuming various channel and data estimation techniques at the BS. 

In~\cite{Risi2014}, the 1-bit quantized massive MIMO with QPSK signaling is considered. With least-square (LS) channel estimator and maximum ratio combining (MRC) or zero-forcing (ZF) receiver, it is shown that high data rates can be attained. In~\cite{jacobsson2015} two efficient near maximum-likelihood channel estimator and data detector are developed which allow the support of multiuser and the use of high-order modulation transmission. As well, the authors in~\cite{jacobsson2015} show that LS-channel estimation combined with MRC receiver suffices for holding high-order modulation transmission in multiuser massive MIMO scenario with 1-bit quantization. In~\cite{Wen2015} pilot and data are jointly utilized for channel estimation in a single-cell with 1-bit ADCs, where it is indicated that this approach outperforms the pilot-only approach. Still, its performance under the multicell case is not addressed and no bounds on capacity are given. Results in~\cite{Wang2016} show that a mixed structure of 1-bit ADCs and conventional ADCs can achieve higher spectral and energy efficiency than the traditional massive MIMO. In~\cite{Mezghani2017}, the channel sparsity in massive MIMO channel is exploited to estimate the channel blindly using the expectation-maximization algorithm. It is demonstrated that reliable channel estimation is still possible with 1-bit ADCs. 

In~\cite{7600443}, the uplink in a wideband massive MIMO channel with 1-bit ADCs is considered. With the assumption of quantization noise (QN) being independent and identically distributed (i.i.d.), linear minimum mean square error (LMMSE) channel estimator is derived. Further, with employing MRC or ZF receiver, a lower bound on the achievable rate is obtained in~\cite{7600443}. Moreover, it is proven that the assumption of i.i.d. QN becomes increasingly accurate when working in the low signal-to-noise ratio (SNR) regime or when the number of channel taps is sufficiently large.  In~\cite{Li2016}, a more accurate model which captures the temporal/spatial correlation of QN is derived by utilizing the Bussgang theorem~\cite{Bussgang52}, while assuming a single-cell case with 1-bit quantization. It is shown that considering correlation among QN samples can further improve channel estimate, especially when the input signal is correlated. In Rayleigh fading channels, it turns out that the assumption of i.i.d. QN serves as a good approximation when operating in the low-SNR regime or when the number of users is sufficiently large, validating the observations in~\cite{7600443}. Grounded on the approximate lower bounds on achievable rate established in~\cite{Li2016} for MRC and ZF receivers, it is concluded that high spectral efficiency can be achieved in the single-cell scenarios. In~\cite{Jacobsson2017}, the authors generalize the 1-bit quantized model in~\cite{Li2016} to an arbitrary number of quantization bits. Therefore, they show that high-order modulation is possible even under the 1-bit quantization case, while with a few bits, one can approach the rate achieved when no quantization is used.

Previous works on the 1-bit quantized massive MIMO have focused on \emph{time-multiplexed pilot} (TP) scheme~\cite{1193803,4176578}, where pilot and data symbols are orthogonal in the time domain. In TP scheme, pure pilot symbols of length, say $\tau$ symbols, are transmitted by all users at the start of transmission for channel estimation at the BS. Then, data communication takes place during $T-\tau$ symbol intervals, where $T$ is the coherence time of the channel over which the channel is assumed constant (see the illustration in Fig.~\ref{fig:framestructure}(a)). An interesting and well-known technique called \emph{superimposed pilot} (SP)~\cite{987005,1232493,1658222}, sends pilot and data symbols side-by-side (see the illustration in Fig.~\ref{fig:framestructure}(b)). Unlike SP scheme, with TP scheme no data is transmitted during the training phase, thus a decrease in spectral efficiency might be incurred, especially when $T$ is short. For the \emph{unquantized} massive MIMO channel, it was shown in~\cite{7865983} that SP has a potential to achieve higher spectral efficiency than TP, especially when advanced signal processing is used at the BS.

The authors in~\cite{2017arXiv170907722V} consider the conventional (unquantized) multicell massive MIMO system with SP approach and derive a closed-form expression for the ergodic achievable rate in the uplink. In that respect, it is shown that when both SP and TP approaches are optimized, comparable data rates can be obtained in practical multicellular scenarios. This is because pilot contamination resulting in TP can be, in some sense, equally bad as data interference in SP due to simultaneous transmission of pilots and data. With more advanced signal processing, it is believed that SP has the potential to outperform TP. The direct extension of the work in~\cite{2017arXiv170907722V} to the quantized channel is not possible due to the presence of the QN (not independent of the input signal) which enters in many places in the analysis, rendering the analysis intractable.

\begin{figure}
	\centering
	\includegraphics[width=0.35\textwidth]{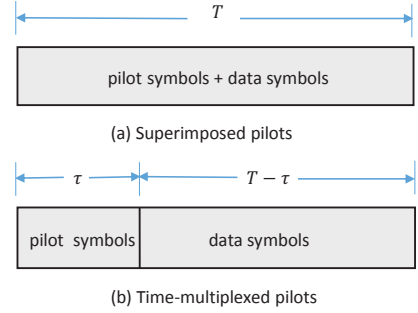} \\
\caption { Illustration of the difference between superimposed pilots and time-multiplexed pilots transmission schemes.}
  \label{fig:framestructure}
\end{figure}

\textit{Our contributions}: In this work, we consider a power-controlled 1-bit quantized massive MIMO system with SP scheme. We are interested in the uplink achievable rate of such a system while assuming LMMSE channel estimator and MRC receiver for data detection at the BS. According to our knowledge, this is perhaps the first attempt that focuses on superimposed pilots in quantized massive MIMO systems. An important contribution of this paper is that the quantized and unquantized systems are asymptotically equivalent as the number of BS antennas $M\to \infty$. Throughout this manuscript we shall refer to the 1-bit quantized channel with SP scheme as \emph{quantized SP} (QSP) and its unquantized (i.e., with infinite-resolution ADCs) counterpart as \emph{unquantized SP} (UQSP).


We commence our work by considering a single-cell scenario, then we extend the results to the multicell scenario. Our theoretical findings agree with the numerical results to corroborate that regardless of the very coarse quantization, QSP provides high data rate. Due to the scope limitation of this work and the lack of an explicit formula of the achievable rate for the 1-bit quantized multicell massive MIMO system with TP, the comparison with the 1-bit quantization with TP case, dubbed \quotes{QTP}, is obtained numerically, where no optimization over training duration or power is considered. The simulation results indicate that QSP can outperform the non-optimized QTP approach in most times. Despite this, a fair evaluation of QSP against the optimized QTP needs more investigations, considering optimizing QTP and using different techniques for channel and data estimation, which is a tedious task and beyond the purpose of this work and hence left for future study.  

The main contributions of our work are summarized as follows:

\begin{enumerate}[leftmargin=*]
	\item A closed-form expression of the LMMSE channel estimate and its mean-square error (MSE) are derived for the QSP. A similar closed-form expression under the no-quantization case is obtained as a special case of the 1-bit QSP.
	\item We obtain a closed-form approximation on the achievable rate for the 1-bit QSP and the optimal power allocation among pilot and data is also given, where the results turn out to be accurate when working in the low-SNR regime or when the number of users is large. A lower bound on the achievable rate for the infinite-resolution case is recovered as a special case.
	\item It is shown that regardless the coarsely quantized signal and symbol-by-symbol superposition of pilot and data, it is even possible to achieve high data rates in multicell massive MIMO.
	\item We show numerically that, under 1-bit quantization case, removing the pilot contribution utilizing the estimated channel from the quantized signal incurs a negligible loss in information in the regimes of low-SNR and a large number of users.
	\item We show that the performance gap between the quantized and unquantized systems vanishes when the number of BS antennas is asymptotically large, i.e., QN can be averaged out asymptotically. In the multicell case, the data rates for QSP and UQSP saturate and converge to a fixed value, given by $\log (1+ \alpha T /\zeta)$ bits/s/Hz, where $\alpha \in (0,1)$ is the power fraction allocated to pilot and $\zeta>K$ is a constant which depends on the number of users and geometry of the network.
\end{enumerate}

\textit{Outline of the paper}: In Sec.\ref{sec:sec_2} the signal model for the 1-bit quantized single-cell MIMO with SP is presented. In Sec.\ref{sec:sec_3} a linear modeling of the quantizer and the details of channel estimation are given. In Sec.\ref{sec:sec_4} the analysis of achievable rates for the quantized and unquantized systems is shown. In Sec.\ref{sec:sec_5}, we extend the previous results to the multicell case and some asymptotic results are established. In Sec.\ref{sec:sec_6}, some numerical results for a multicell massive MIMO system are presented to validate the analytical results and Sec.\ref{sec:sec_7} summarizes this paper.


\section{System Model for single-cell Massive MIMO}
\label{sec:sec_2}

In the first part of this work, we consider the uplink of a single-cell massive MIMO system where the BS has $M$ antennas and serves $K$ ($K \ll M$) single-antenna users in the same time-frequency resource. We consider a block-flat Rayleigh fading channel which remains constant over $T$ symbol intervals, i.e., $T$ is the coherence time of the channel. The channel gain between the $m$-th BS antenna and user $k$ is represented by $\sqrt{\beta_k} h_{mk}$, where $\beta_{k}$ and $h_{m k}$ are the large-scale and small-scale fading coefficients, respectively. The coefficients $\{h_{m k}\}$ are assumed to be i.i.d. $\mathcal{CN}(0,1)$. We assume that the distance between user $k$ and the BS is sufficiently large so that $\beta_{k}$ becomes constant over the antenna array.

In Fig.~\ref{fig:one-bit system} (a) a simplified illustration of the uplink system is shown where the $m$-th BS antenna is equipped with a pair of 1-bit ADCs for both the in-phase and quadrature components of the discrete-time baseband signal $y_m[t], t=1,\cdots,T$. With superimposed pilots, $y_m[t]$ is a noisy superposition of pilot and data symbols subject to channel distortion. For ease of analysis, we assume the transmission of pilot and data symbols takes place over $T$ symbol intervals, i.e., pilot and data have the same length $T$. So, during one coherence interval, the BS collects $2MT$ binary samples. We assume the BS first estimates the channel using LMMSE, then uses the channel estimate for data estimation using MRC receiver. For pilots, we consider all $K$ users use mutually orthogonal sequences. 

With quantized channels, it is shown in~\cite{Ivrlac2007} that the structure of pilot sequences affects the performance. Lately,  the authors in~\cite{7600443} show that sparse pilot sequences (i.e., with many zero entries) give rise to degradation of channel estimation performance. Thus, under quantized channels, non-sparse pilot sequences in the time domain are more preferable. To meet such a requirement, without loss of generality, we assume that the pilot sequences are drawn randomly from a Fourier basis matrix. We recall here that such pilot sequences have nonzero entries where each entry has an absolute value of one.
\begin{figure}
	\centering
	\includegraphics[width=0.45\textwidth]{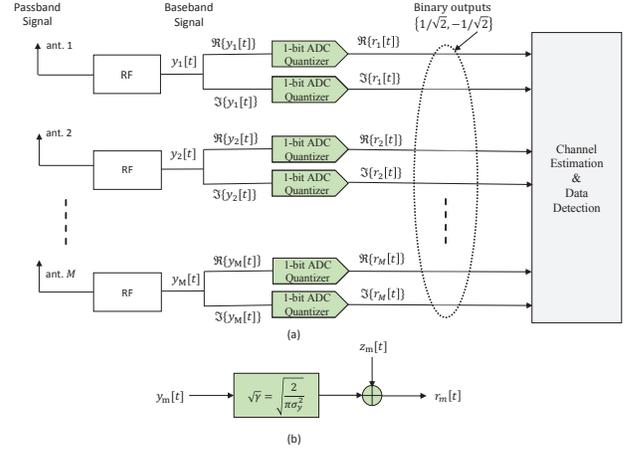} \\
	\caption {Massive MIMO with 1-bit ADCs (a) A schematic representation of uplink system structure with 1-bit ADCs (b) Model of 1-bit ADC quantizer using the Bussgang theorem for Gaussian input, where $y_m[t] \sim \mathcal{CN}(0,\sigma_y^2)$, variance of QN is $\sigma_z^2 = 1-\frac{2}{\pi}$ and $\E  \{z_{m}^{\ast}[t] y_m[t]  \} = 0$.} 
	\label{fig:one-bit system}
\end{figure}

Let $\mf x_k = \left ( x_{k}[1], x_{k}[2], \cdots,x_{k}[T] \right )^{\TT}$ be a vector of $T$ symbols sent by user $k$ during one coherence interval satisfying power constraint:
\begin{equation} \label{eq: power constraint}
	\E \left \{ \norm{\mf{x}_{k}}^2 \right \} \le T.
\end{equation}
Further, let $c_{k}[t]$ and $s_{k}[t]$ be the pilot and data symbols associated with user $k$ during time instant $t$, respectively. Then $x_{k}[t]$ is written as the superposition of $c_{k}[t]$ and $s_{k}[t]$:
\begin{equation} \label{eq:pilot_data_SP}
	x_{k}[t] =\sqrt{\alpha_k}  {c}_{k}[t] + \sqrt{\bar \alpha_k }   s_{k}[t],
\end{equation}
where $\alpha_{k}$ and $\bar \alpha_{k}$ are the power fractions allocated to pilot and data symbols, respectively, such that $ \alpha_{k} + \bar \alpha_{k} = 1$. According to our choice of pilot sequences, we have $\sum_{t=1}^{T} |c_{k}[t]|^2=T$, therefore, to meet the power constraint~\eqref{eq: power constraint}, we assume $\E\{|s_{k}[t]|^2\} = 1$.

In cellular communication systems, power control in the uplink is one of the most important mechanisms for fairness provision among different users experiencing different fading conditions. In particular, the imbalance in received powers of different users becomes a critical issue in 1-bit quantized channels due to the saturation effect of ADC. For instance, when the received signal strength of a close-to-BS user overwhelms the signal of a cell-edge user, the output of 1-bit ADC would have insufficient information (if not lost) about the weak user, making the detection of their presence hard. In this work, power control based on \textit{statistical channel-inverse} is assumed~\cite{Li2016}~\cite{2017arXiv170907722V}. Thus, in Rayleigh fading, the transmit power of user $k$, denoted by $\rho_k$, is chosen to be
\begin{equation}\label{eq:power control}
	\rho_k = \frac{\rho}{\beta_k},
\end{equation}
such that the average received power by the BS from each user is the same, given by $M \rho$, where $\rho$ is some fixed power requirement. As a result, we have $\alpha_k = \alpha$ and $\bar \alpha_k = \bar \alpha, k=1,2,\cdots, K$.

Based on the above discussion, $y_m[t]$ can be written as
\begin{equation}
	y_{m}[t] = \sum_{k=1}^{K}  \sqrt{ \rho} h_{m k} \left(\sqrt{\alpha} c_{k}[t] + \sqrt{\bar \alpha} s_{k}[t]\right) +  w_{m}[t],
	\label{eq:BB signal scalar}
\end{equation}
where $w_{m}[t]\sim \mathcal{CN}(0,1)$ is additive white Gaussian noise (AWGN) which is assumed i.i.d. across space and time. The quantized signal of~\eqref{eq:BB signal scalar} with zero-threshold 1-bit ADCs is:
\begin{equation}
	\label{orig_quantized}
	r_{m}[t]  =\frac{1}{\sqrt{2}} \mathrm {sign} \left( \Re \{ y_{m}[t] \} \right ) + \frac{j}{\sqrt{2}}\mathrm {sign}\left( \Im \{ y_{m}[t]  \} \right),
\end{equation}
where $\Re\{\cdot\}$ and
$\Im\{\cdot\}$ denote the real and complex parts of a complex quantity, respectively,  $\mathrm {sign}(\cdot)$ is the signum function which returns $1$ or $-1$ when its argument is greater or less than zero, respectively and $1/\sqrt{2}$ is a scale normalization factor. It follows that $r_{m}[t] \in \mathcal{A}\triangleq \frac{1}{\sqrt{2}}\{1\pm j,-1\pm j\}$.

Equation~\eqref{eq:BB signal scalar} can be written in a compact matrix form. Let $\mf s_{k} =  \left(s_k[1], \cdots,s_k[T] \right)^{\TT},\mf c_{k} =  \left(c_k[1], \cdots,c_k[T] \right)^{\TT}, \mf h_{k} =  \left(h_{1 k}, \cdots, h_{M k}\right)^{\TT}, \mf y[t]  =  \left(y_{1}[t],\cdots,y_{M}[t]\right)^{\TT}$ and $\mf w[t]  =  \left(w_{1}[t] ,\cdots, w_{M}[t]\right)^{\TT}$ be the length-$T$ vectors of data and pilot symbols, channel vector between user $k$ and all BS antennas, the received signal and AWGN during time $t$ at all BS antennas, respectively. Similarly, we denote by  $\mf r[t]  =  \left(r_{1}[t],\cdots,r_{M}[t]\right)^{\TT}$ and $\mf z[t]  =  \left(z_{1}[t],\cdots,z_{M}[t]\right)^{\TT}$  the received quantized signal and QN across all BS antennas during time instant $t$, respectively. Rearranging~\eqref{eq:BB signal scalar} for all $m=1,2,\cdots,M$ and $t=1,2,\cdots,T$, the $M \times T$ received signal can be written as
\begin{equation} \label{eq:BB signal matrix form}
	\mf Y = \sqrt{\alpha \rho} \mf H \mf C +  \sqrt{\bar \alpha \rho} \mf H \mf S + \mf W,
\end{equation}
and hence the corresponding quantized signal is
\begin{equation}
	\label{eq:orig_quantized_matrix}
	\mf R = \frac{1}{\sqrt{2}} \mathrm {sign} \left( \Re\{\mf Y\}\right) + \frac{j}{\sqrt{2}} \mathrm {sign} \left( \Im\{\mf Y\} \right),
\end{equation}
where the function $\mathrm {sign} (\cdot)$ is applied element-wisely on a matrix,  $\mf Y = \left( \mf y[1],\cdots,\mf y[T]\right), \mf R = \left( \mf r[1],\cdots,\mf r[T]\right)$, $\mf W=\left( \mf w[1],\cdots,\mf w[T]\right)$  and $\mf H , \mf C$ and $\mf S$  are the composite channel, pilot and data matrices, defined, respectively, as $\mf H = \left (\mf h_{1},\cdots,\mf h_{K} \right)\in \mathbb{C}^{M \times K}, \mf C = \left (\mf c_{1}, \cdots, \mf c_{K}\right)^{\TT}\in \mathbb{C}^{K\times T}$ and $\mf S = \left (\mf s_{1},\cdots, \mf s_{K}\right)^{\TT} \in \mathbb{C}^{K \times T}$.

\section{Linear model and channel estimation}
\label{sec:sec_3}
The model equation~\eqref{orig_quantized} is non-linear which makes the analysis intricate. Yet, according to the Bussgang theorem~\cite{Bussgang52}, if the quantizer's input is Gaussian, its output can be decomposed into a sum of a scaled version of its input and uncorrelated QN, where the QN itself is generally correlated~\cite{7600443}. Despite this, the decomposition involves the nonlinear \emph{arcsine law} which makes the nonlinearity unavoidable, i.e., see~\cite{Mezghani2012}~\cite{Li2016} for more details. To avoid such nonlinearity, we make simplifications to the model which allow us to get a closed-form expression for the channel estimator and uplink achievable rate.  

\subsection{Assumptions} \label{sec: Assumption}
To proceed with the development of theoretical results, we allow the following assumptions:
\begin{enumerate}[leftmargin=*]
	\item All data symbols $\{s_k[t]\}$ are i.i.d. $\mathcal{CN}(0,1)$. This allows using existing lower-bounding techniques to obtain a lower bound on capacity~\cite{1193803}.
	\item The signals $\{y_m[t]\}$ are Gaussian random variables, each with zero-mean and variance $\sigma_{y}^2= K \rho+1$. This is justified by the central limit theorem (CLT) as the distribution of $y_m[t]$ tends towards Gaussian as $K$ gets larger. For $K$ being sufficiently large should come as no surprise in massive MIMO.
	\item In the regime of low-SNR or large-$K$, the QN samples are approximated as i.i.d. random variables~\cite{7600443}~\cite{Li2016}.
\end{enumerate}
\subsection{ The linear model  of 1-bit quantizer}
Here we apply the Bussgang decomposition to obtain a linear model for the quantizer, which will be utilized for analysis.  

Let $\bar {\mf {y}}_m=(y_m[1],\cdots,y_m[T])^{\TT},\bar{\mf {r}}_m=(r_m[1],\cdots,r_m[T])^{\TT}$, $\bar{\mf z}_m = (z_m[1],\cdots,z_m[T])^{\TT}$ and $\bar{\mf w}_m=(w_m[1],\cdots,w_m[T])^{\TT}$ be the column vectors constructed from the $m$-th rows of $\mf Y, \mf R, \mf Z$ and $\mf W$, respectively. Further, we define $ \bar{\mf h}_m=(h_{m1},\cdots, h_{mK})^{\TT}$ as the column vector of the channel from all $K$ users to the $m$-th BS antenna. We use bars over all symbols to distinguish them from the original columns of corresponding matrices. Thus, we can write $\bar {\mf {y}}_m=\sqrt{\alpha \rho}  \mf C^{\TT} \bar{\mf h}_m + \sqrt{\bar \alpha \rho} \mf S^{\TT} \bar{\mf h}_m + \bar{\mf w}_m$ and $\bar {\mf {r}}_m= (\mathrm {sign} \left( \Re\{\bar {\mf {y}}_m\}\right) + j\mathrm {sign} \left( \Im\{\bar {\mf {y}}_m\} \right))/{\sqrt{2}} $. Under Rayleigh fading, the rows of $\mf Y$ (and hence the rows of $\mf R$) are i.i.d. with a common covariance matrix defined as $\Sigma_{\bar{ \mf  y}_m} \triangleq  \E \left \{ \bar{ \mf  y}_m \bar{ \mf  y}_m^{\HH} \right \}$, given by
\begin{equation} \label{eq:input_covariance}
	\Sigma_{\bar {\mf {y}}_m} = \alpha \rho \mf C^{\TT} \mf C ^{\ast}+ (\bar \alpha \rho K + 1) \mf I_{T}.
\end{equation}
Thus it suffices to focus only on the linear model of an arbitrary row $m$ of $\mf R, \bar {\mf r}_m$.

Since $\bar {\mf {y}}_m$ is jointly Gaussian (see Assumption 2), therefore, according to the Bussgang theorem, $\bar{\mf r}_m$ can be written as $\bar{\mathbf {r}}_{m} = \mathbf B \bar {\mathbf {y}}_{m}  + \bar{\mathbf z}_{m}$, where $\mathbf B$ is some $T\times T$ matrix, chosen such that $\bar{\mathbf z}_i$ is uncorrelated with $\bar {\mathbf {y}}_{m}$, i.e.,
\begin{equation}\label{qnoise_inp_crosscorrelation}
	\E  \{z_{m}^{\ast}[t] y_{m }[t^{\prime}] \} = 0,\quad \forall t,t^{\prime} \in \{1,\cdots,T\}.
\end{equation}
From~\cite{Li2016}, it is easy to show that $\mf B$ admits the simple form $\mf B= \sqrt{\gamma} \mf I$, where $\gamma$ is a scaling factor defined as
\begin{equation}
	\gamma \triangleq \frac{2}{\pi\sigma_{y}^2},
\end{equation}
and hence $\bar{\mathbf {r}}_{m}$ takes the following form:
\begin{align}\label{eq: quantized_sig_linearmodel2}
	\bar{\mathbf {r}}_{m} = \sqrt{\gamma} \bar {\mathbf {y}}_{m}  + \bar{\mathbf z}_{m}.
\end{align}

From~\eqref{eq: quantized_sig_linearmodel2}, the covariance matrix of QN, defined as $\Sigma_{\bar{ \mf  z}_m} \triangleq  \E  \left \{ \bar{ \mf  z}_m \bar{ \mf  z}_m^{\HH} \right \}$, is given by $\Sigma_{\bar{ \mf  z}_m} = \Sigma_{\bar{ \mf  r}_m}  -  {\gamma} \Sigma_{\bar{ \mf  y}_m}$,
where $\Sigma_{\bar{ \mf  r}_m}\triangleq  \E  \left \{ \bar{ \mf  r}_m  \mf {\bar{ \mf  r}}_m^{\HH} \right \}$ is the autocorrelation matrix of quantizer's output. It is well-known that when quantizer input is Gaussian, $\Sigma_{\bar{ \mf  r}_m}$ follows the arcsine law~\cite{272490}. However, working with the nonlinear arcsine operator turns out to be intractable. However, with the assumption of i.i.d. QN (Assumption 3), $\Sigma_{\bar{ \mf  z}_m} $  reduces to a diagonal matrix given by~\cite{Li2016}:
\begin{equation} \label{eq: noise_covariance}
	\Sigma_{\bar{ \mf  z}_m} \approx \sigma_z^2  \mf  {I}_{T},
\end{equation}
where $\sigma_z^2=1-{2}/{\pi}$ is the variance of QN samples.  Figure.~\ref{fig:one-bit system}(b)  shows a schematic representation of the equivalent linear model of 1-bit quantizer.

It is shown in~\cite{7600443}~\cite{Li2016} that with TP scheme, QN is not only uncorrelated with the input signal but also with the channel. This useful result is stated in the following lemma.
\begin{lemma}[\cite{7600443}~\cite{Li2016}]
	\label{lemm:ch_QN_Correlation}
	For any $m,m^{\prime}\in \{1,\cdots,M\}$ and any $t \in 1,2,\cdots,T$, we have
	\begin{equation}\label{eq: z_h_crosscorrelation}
		\E  \{h_{m^{\prime} k}  z_{m}^{\ast}[t] \} = 0,\quad t=1,\cdots,T.
	\end{equation}
\end{lemma}

It should be noted that Lemma~\ref{lemm:ch_QN_Correlation} is still valid when SP scheme is used. We will use Lemma~\ref{lemm:ch_QN_Correlation} when deriving the LMMSE channel estimator.

\subsection{ Channel estimation}
For QSP, we use~\eqref{eq: quantized_sig_linearmodel2} to derive the LMMSE channel estimate of the channel between user $k$ and the $m$-th BS antenna, $ h_{m k} $. With UQSP, the LMMSE channel estimate can be recovered as a special case of the result of QSP when $\bar{\mathbf {r}}_{m}=\bar{\mathbf {y}}_{m}$  (i.e,  $\sigma_z^2=0$ and $\gamma=1$). By left-multiplying $\bar{\mathbf {r}}_{m} $ by  $\mf c_{k}^{\HH}$ we get the following single-valued function:
\begin{equation}\label{eq: CE_model}
	u_{m k} =T \sqrt{\alpha \rho \gamma} h_{m k} + v_{m k},
\end{equation}
where $ v_{m k}$ is a non-Gaussian effective noise defined as
\begin{equation}
	\label{eq: v_mk}
	v_{m k}  = \sqrt{\bar \alpha \rho \gamma}  \mf {c}_{k}^{\HH} \mf {S}^{\TT} \bar {\mf{h}}_m + \sqrt{\gamma}  \mf c_{k}^{\HH} \mf{\bar w}_{m} +  \mf c_{k}^{\HH} \mf{\bar z}_{m}.
\end{equation}

Based on~\eqref{eq: CE_model}, the channel estimate and variance of estimation error are given in the following lemma.
\begin{lemma}
	\label{lemma:QSPS}
	Consider QSP with i.i.d. QN, the LMMSE estimate of $h_{m k}$, denoted by $\hat{h}_{m k} $, is
	\begin{equation}\label{eq: CE_QSP}
		\hat{h}_{m k} =\frac{\sqrt{\alpha \rho \gamma}}{\alpha \rho \gamma T +\bar \alpha \rho \gamma K  + \gamma +\sigma_z^2 } u_{m k}
	\end{equation}
	and the variance of estimation error is
	\begin{equation} \label{eq: CEE_QSP}
		\sigma_{\tilde{h}}^2 = 1-\frac{{\alpha \rho }}{\alpha \rho   +\frac{\bar \alpha \rho K  + \sigma_z^2/\gamma +{1}}{T} }
	\end{equation}
\end{lemma}
\goodbreak
\begin{proof}
	Using~\eqref{eq: CE_model} and applying the standard LMMSE solution~\cite{sayed2011adaptive} for $h_{mk}$, we get $\hat{h}_{m k}= \frac{ \E \{ h_{m k } u_{m k}^{\ast} \}}{ \E \{|u_{m k}|^2\}} u_{m k}$. In calculating $\E \{ h_{m k } u_{m k}^{\ast} \}$ we make use of Lemma~\ref{lemm:ch_QN_Correlation} and in calculating $\E \{|u_{m k}|^2\}$ we make use of~\eqref{eq:input_covariance},~\eqref{qnoise_inp_crosscorrelation} and~\eqref{eq: noise_covariance}.
	
\end{proof}
The channel estimate and the variance of estimation error of UQSP can be recovered from Lemma~\ref{lemma:QSPS}. Thus we are led to the following lemma:
\begin{lemma}\label{lemma:UQSP}
	For UQSP, the LMMSE estimate of $h_{m k}$, denoted by $\hat{h}_{m k} $, is given by
	\begin{equation}\label{eq: CE_UQSP}
		\hat{{h}}_{m k} =\frac{\sqrt{\alpha \rho}}{\alpha T \rho + \bar \alpha K \rho  + 1  } \bar{u}_{m k}
	\end{equation}
	and the variance of estimation error is
	\begin{equation}\label{eq: CEE_UQSP}
		\bar{\sigma}_{\tilde{{h}}}^2 = 1- \frac{{\alpha \rho}}{\alpha \rho   +\frac{\bar \alpha \rho K+1}{T} }.
	\end{equation}
	where $\bar{u}_{m k}=  T \sqrt{\alpha \rho } h_{m k} +\sqrt{\bar \alpha \rho } { \mf {c}}_{k}^{\HH} \mf {S}^{\TT} \bar {\mf{h}}_m + \mf c_{k}^{\HH} \mf{\bar w}_{m}$.
\end{lemma}
\goodbreak
\begin{proof}
	Note that under UQSP, $\mf{\bar r}_{m} = \mf{\bar y}_{m}$ (i.e., $\sigma_z^2 = 0$ and $\gamma=1$). Thus by redefining~\eqref{eq: CE_model}-~\eqref{eq: CEE_QSP} according to this change,~\eqref{eq: CE_UQSP} and~\eqref{eq: CEE_UQSP} follow immediately.
\end{proof}

Although $\hat{h}_{m k}$ in Lemma~\ref{lemma:QSPS} doesn't coincide with the MMSE solution ($\hat{h}_{m k}$ is not truly Gaussian because $u_{m k}$ is not Gaussian), $\hat{h}_{m k}$ can still be approximated as Gaussian. This elicits from the fact that $\hat{h}_{m k}$ is a weighted sum of $T$ random variables (typically large), thus by the virtue of CLT, the distribution of $\hat{h}_{m k}$ tends towards Gaussian.

By inspecting Lemmas~\ref{lemma:QSPS} \&~\ref{lemma:UQSP}, the variance of estimation error can be arbitrarily small as $T$ increases. It is interesting to see from~\eqref{eq: CEE_QSP} that the effect of AWGN and QN is scaled down by a factor of $T$, compared with a factor of $\tau$ with TP scheme (see~\cite{Li2016}), where $T$ is typically much larger than $\tau$. Still, with SP, there is an extra term $\bar \alpha \rho \gamma K/T$ due to data interference while there is no data interference during the training phase under the TP scheme. In both the quantized and unquantized systems, it is worth observing that there is a \emph{saturation effect} in LMMSE performance as $\rho$ grows large. This suggests that there is a $\rho$-threshold after which no increase in data rate is achieved. Increasing pilot power (by increasing $\alpha$) leads always to an improvement of channel estimate quality while data rate may decrease as data power decreases, even so, it turns out that this relation is not monotonic. Hence, there is a trade-off between channel estimate quality and data rate, implying the existence of an optimal fraction $\alpha$ of total power which gives rise to the highest data rate.    

We conclude this section by making the following remark on the fundamental challenges in the analysis of QTP.

\begin{remark}\label{rem:CE_data_depednecy}
	In contrast with QTP where the channel estimate is independent of data, with QSP the channel estimate (and hence the estimation error) is not only dependent on data symbols transmitted from all users but also on AWGN and QN. Besides, although the QN and input signal are uncorrelated, they are not independent. The QN enters different places in the analysis, especially when calculating the variance of MRC output. Typically, this requires a complete characterization of statistical relationships with other random variables through the joint probability densities, which are hard to obtain. Altogether, this is an undesirable property which renders the capacity analysis intractable as will be discussed next.
\end{remark}

\section{Analysis of achievable rates}
\label{sec:sec_4}
This section is concerned with obtaining an approximate lower bound on the ergodic achievable rate for QSP massive MIMO where the BS uses LMMSE and MRC for channel and data estimation, respectively. We then leverage the analysis of QSP to obtain a lower bound on the ergodic achievable rate for UQSP. In the rest of this section, a comparison between QSP and UQSP is established through asymptotic analysis.

Without loss of generality, we focus on the information rate of the $k$-th user during an arbitrary time $t$.  Thus the channel model corresponds to the $t$-th column of $\mf R$ which can be written as $\mf r[t] =\sqrt{\gamma} \mf y[t]+\mf z[t]= \sqrt{\alpha \rho \gamma} \mf H \bar{ \mf c}[t] + \sqrt{\bar \alpha \rho \gamma} \mf H \bar{ \mf s}[t] + \sqrt{ \gamma}  \mf w[t] + \mf z[t]$, where $\bar{ \mf s}[t]=(s_{1}[t], s_{2}[t], \cdots, s_{K}[t])^{\TT}$ and $\bar{ \mf c}[t]= (c_{1}[t], c_{2}[t], \cdots, c_{K}[t])^{\TT}$ are the column vectors of $\mf S$ and $\mf C$, respectively. We stack all channel estimates $\{\hat h_{m k}\}$ of the $k$-th user as a column vector given by $\hat{\mf {h}}_k =(\hat h_{1 k},\cdots,\hat h_{m k},\cdots, \hat h_{M k})^{\TT}$. Thus we have
\begin{equation}\label{eq: h_k_hat}
	\hat{ \mf  {h}}_k = \xi  \mf R \mf c_{k}^{\ast},
\end{equation}
where $\xi \triangleq {\sqrt{\alpha \rho \gamma}}/({\alpha \rho \gamma T +\bar \alpha \rho K \gamma + \gamma +\sigma_z^2 })$.

Given the channel estimate $\hat {\mf H}  = [\hat{\mf  {h}}_1 \hat{\mf  {h}}_2 \cdots \hat{\mf  {h}}_K ]\in \mathbb{C}^{M\times K}$, the contribution of pilots can be removed from ${\mf  r}[t]$ as follows:
\begin{IEEEeqnarray} {rCl}
	\label{eq:PR_model} \nonumber
	&& { \mf  r}^{\text{PR}}[t] \triangleq  { \mf  r}[t]-\sqrt{\alpha \rho \gamma} \hat {\mf H}  \bar{ \mf  c}[t]\\
	&&= \sqrt{\bar \alpha \rho \gamma}  {\mf H}  \bar{ \mf  s}[t] +\sqrt{\gamma} { \mf  w}[t] +  { \mf  z}[t] + \underbrace{\sqrt{\alpha \rho \gamma} \tilde {\mf H}  \bar{ \mf  c}[t]}_{\text{residual pilot noise}}
\end{IEEEeqnarray}
where $\tilde {\mf H} = \mf H-\hat {\mf H} $ is the matrix of channel estimation error. In~\eqref{eq:PR_model} we use the notation \quotes{PR} to denote pilot removal. The basic problem of the model~\eqref{eq:PR_model} is that the last term is not independent of the other remaining terms (also cross-correlated), especially data (see Remark~\ref{rem:CE_data_depednecy}), which is a primary source of intractability in the analysis of capacity.

To maintain the analysis tractable, pilot removal (PR) is not considered in this work, similar to the technique used in~\cite{2017arXiv170907722V} for the analysis of the unquantized MIMO channel with SP. Instead, $\mf r[t]$ is utilized directly at BS for data estimation, which can be rewritten as 
\begin{align}
	\label{sig_with_pilot}
	{ \mf  r}[t] = \underbrace{\sqrt{\bar \alpha \rho \gamma} \mf  h_k s_{k}[t] }_{\substack{\text{$k$-th user} \\ \text{information}}}+ \underbrace{\tilde{ \mf  r}[t]}_{\substack{\text{non-Gaussian} \\ \text{effective noise}}},
\end{align}
where $\tilde{ \mf  r}[t] \in \mathbb{C}^{M}$ is given by
\begin{IEEEeqnarray}{rcl} \nonumber
	\tilde{ \mf  r}[t] &=& \sqrt{\bar \alpha \rho \gamma} \left( \mf  H \bar{ \mf s}[t] -  \mf  h_k s_{k}[t]\right )+ \sqrt{ \alpha \rho \gamma}  {\mf H \bar{ \mf c}[t]}\\
	&& \qquad +\> \sqrt{\gamma} { \mf  w}[t] + { \mf  z}[t].
	\label{eq:non-Gaussian effective noise}
\end{IEEEeqnarray}

Although PR gives rise to increase in data rate, it turns out that the loss in information resulting from using~\eqref{sig_with_pilot}, instead of~\eqref{eq:PR_model} is not large, especially in the low-SNR and large-$K$ regimes, where these two regimes come at no surprise in massive MIMO settings. From~\eqref{eq: h_k_hat} and~\eqref{sig_with_pilot}, the MRC output, $\hat s_{k}[t]  \triangleq \frac{1}{{M}}\hat{ \mf  h}_k^{\HH} \mf  r [t]$, can be written as
\begin{align} \label{eq:process_sig}
	\hat s_{k}[t]  = \underbrace{\frac{\xi \sqrt{\bar \alpha \rho \gamma}}{{M}}  \mf {c}_k^{\TT}  \mf  R^{\HH} \mf  h_k s_{k}[t]}_{\text{useful information}}   + \underbrace{\frac{\xi }{{M}}  \mf {c}_k^{\TT}  \mf  R^{\HH} \tilde{ \mf  r}[t]}_{\text{effective noise}},
\end{align}
where the scaling factor $1/M$ is introduced here for establishing some asymptotic results later. 

While pilot removal is not considered, however, the model~\eqref{eq:process_sig} is still intractable. This elicits from the fact that the quantized signal $\mf R$ is not independent of information $s_{k}[t]$. In addition, the distribution of effective noise in~\eqref{eq:process_sig} is unknown and the effective noise is not independent of the signal part (also correlated), especially due to the presence of QN (see also Remark~\ref{rem:CE_data_depednecy}). Despite this, we make use of some mathematical simplifications and asymptotic results, owing to some statistical properties of QN to tackle the intractability of the model~\eqref{eq:process_sig}.

In the following, we leverage~\eqref{eq:process_sig}  to derive an approximate lower bound on the mutual information $I(s_{k}[t];\hat s_{k}[t])$ with the Gaussian assumption on $s_{k}[t]$.

\subsection{Quantized superimposed pilots}
From the aforementioned discussion, the exact capacity of~\eqref{eq:process_sig} is unknown. The following result holds when the number of BS antennas is sufficiently large.
\begin{thm} [approximate lower bound]
	\label{thm: MRC_Rate_QSP}
	Considering QSP massive MIMO where QN is assumed to be i.i.d. (Assumption 3), if the BS employs LMMSE channel estimation and MRC data estimation, an approximate lower bound on the achievable rate in uplink is
	\begin{equation}
		\label{eq: MRC_Rate_QSP}
		R_{\text{QSP}} \approx \log \left( 1+ \Upsilon_{\text{QSP}}  \right) \quad ({\text{bits/s/Hz}}),
	\end{equation}
	where $\Upsilon_{\text{QSP}}$ is the effective SNR at the MRC output given by~\eqref{eq:SNR_QSP_singlecell} which appears at the top of page~\pageref{eq:SNR_QSP_singlecell}.
\end{thm}
\begin{figure*}[!btp]
	\normalsize
	\begin{IEEEeqnarray}{rCl}  \nonumber
		\Upsilon_{\text{QSP}} &=& \biggl({ \bar \alpha \alpha  \rho ^2 T^2 M}\biggl) \biggl/ \biggl (
		\bar \alpha  \rho ^2 K  T M  + \frac{T}{4} \left(8 \alpha  \bar \alpha  \rho ^2-4 \alpha  K^2 \rho ^2+\pi ^2 K^2 \rho ^2-4 \alpha  K \rho +2 \pi ^2 K \rho +\pi ^2\right)\\
		&&+\>  \alpha ^2 K \rho ^2 - 2 \alpha  K \rho ^2 -\frac{1}{4} \pi ^2 K^2 \rho ^2   +  K \rho ^2 - \frac{1}{2} \pi ^2 K \rho  + \alpha  \rho T^2 \left( 1 + K \rho \right)-\frac{\pi ^2}{4}
		\biggl)
		\label{eq:SNR_QSP_singlecell}
	\end{IEEEeqnarray}
	\hrulefill
	\vspace*{4pt}
\end{figure*}
\begin{proof}
	The proof unfolds by rewriting~\eqref{eq:process_sig} as a sum of two parts; desired signal with a multiplicative deterministic channel gain and uncorrelated noise, i.e., $\hat s_{k}[t] = a_{k} s_{k}[t]+  \epsilon_{k} [t]$. Then the lower bound follows by assuming the worst-case noise (i.e., Gaussian noise) with the same variance. 
	The details of the proof and hence the closed-form expression for $\Upsilon_{\text{QSP}}$ in~\eqref{eq:SNR_QSP_singlecell} are shown in Appendix~\ref{app: proof_thm1}.
\end{proof}

The lower bound~\eqref{eq: MRC_Rate_QSP} can be maximized w.r.t.  $\alpha$. The result is stated in the following corollary.
\begin{corollary} \label{cor:optimal_alpha_QSP}
	Let $ \alpha^{\ast}\in(0,1)$ be the optimal value of power fraction which maximizes $R_{\text{QSP}}$~\eqref{eq: MRC_Rate_QSP}. Then $\alpha^{\ast}$ is given by one of the following two roots:
	\begin{equation}
		\alpha^{\ast} = \frac{\lambda \pm \sqrt{\lambda \sigma _y^2 \left(4 \rho  T (T-K)+ (T-1) \sigma _y^2\pi ^2 \right)}}{4 \rho  (K \rho  (T (K+M-T)+1)-T (T-K))}
		\label{eq:opt_alpha_value}
	\end{equation}
	where $\lambda= 4 \rho ^2 K (M T+1)+ \pi ^2 \sigma _y^4 (T-1)$.
\end{corollary}

\begin{proof}
	By the concavity of $\log(1+x)$ and noticing that the first derivative of $R_{\text{QSP}}$ w.r.t. $\alpha$ changes sign due to $R_{\text{QSP}}|_{\alpha=0}=R_{\text{QSP}}|_{\alpha=1}=0$ (i.e., $R_{\text{QSP}}$ is non-monotonic), $\alpha^{\ast}$ can be obtained by solving the quadratic equation $\frac{d}{d\alpha} \Upsilon_{\text{QSP}}=0$ for $\alpha$ and choosing the positive root less than 1.
\end{proof}

\subsection{Unquantized superimposed pilots}
As indicated earlier, the UQSP can be treated as special case of QSP, therefore, we have the following corollary.
\begin{corollary} [lower bound]
	\label{pro:Rate_SP}
	a lower bound on achievable rate under UQSP system when LMMSE and MRC are used for channel and data estimation, respectively, is
	\begin{equation} \label{eq: MRC_Rate_UQSP}
		R_{{\text{UQSP}}} = \log\left ( 1+\Upsilon_{\text{UQSP}}  \right) \quad ({\text{bits/s/Hz}}),
	\end{equation}
	where the effective SNR, $\Upsilon_{\text{UQSP}}$, is given by
	\begin{IEEEeqnarray}{rCl} \nonumber
		&&\Upsilon_{\text{UQSP}} = \Big({\alpha  \bar {\alpha} \rho^2 T^2 M}\Big) \Big/\Big(\bar {\alpha} \rho ^2 K T M  + \alpha  \rho(K \rho+1)T^2\\ \nonumber
		&&+\> 2 \alpha  \bar {\alpha}  \rho^2 T  + \bar {\alpha}  \rho^2 K^2 T + (2-\alpha) \rho K T + T + \bar {\alpha}^2 \rho ^2 K \Big).\\
		\label{eq:SNR_UQSP_singlecell}
	\end{IEEEeqnarray}
	Moreover, the optimal power fraction $ \alpha^{\ast} \in (0,1)$ maximizing $R_{\text{UQSP}}$~\eqref{eq: MRC_Rate_UQSP} is given by one of the following two roots:
	\begin{equation}
		\alpha^{\ast} = \frac{\delta \pm \sqrt{\delta   (  T\rho+1) (K \rho+1)T}}{\rho  \left(K^2 \rho  T+K \left(M \rho  T+\rho -\rho  T^2+T\right)-T^2\right)}
	\end{equation}
	where $\delta = T \left(K^2 \rho ^2+K \rho  (M \rho +2)+1\right)+K \rho ^2$.
\end{corollary}
\begin{proof}
	The proof unfolds by substituting $\gamma = 1, \sigma_z = 0$ and $f(t,t) = (M+1)\sigma_y^2/M$ (due to $\mf r[t] = \mf y[t])$ in~\eqref{eq:noise_var_QSP} to obtain the normalized variance of effective noise $\tilde{\sigma}_{\epsilon_k}^2$, then substituting $\tilde{\sigma}_{\epsilon_k}^2$  in~\eqref{eq:Stnd_lowerbound},~\eqref{eq: MRC_Rate_UQSP} follows. The proof of second part follows the same lines of the proof of Corollary~\ref{cor:optimal_alpha_QSP}.
\end{proof}

Unlike the derivation of~\eqref{eq: MRC_Rate_QSP}, where some approximations are used due to the QN (see the proof of Theorem\ref{thm: MRC_Rate_QSP}), no approximations are required in the derivation of~\eqref{eq: MRC_Rate_UQSP}.

\subsection{Asymptotic analysis}
\label{subsec:Asymptotic analysis_singlecell}
So far, we have obtained the achievable rates for QSP and UQSP systems, however; it seems difficult to obtain a useful mathematical expression of the gap between QSP and UQSP when a general set of system parameters is considered. To get the feeling of this gap, we consider the asymptotic behavior of both systems when the number of BS antennas is asymptotically large. 

We first observe that the common numerator, $\alpha \bar {\alpha} \rho^2 T^2 M$, in $\Upsilon_{\text{QSP}}$~\eqref{eq:SNR_QSP_singlecell} and $\Upsilon_{\text{UQSP}}$~\eqref{eq:SNR_UQSP_singlecell} can be seen as the unnormalized signal power at MRC output which scales with $M$. Second, the denominators (i.e., total noise power) consist of two main kinds of noise: 1) \textit{coherent noise} (first term) which results from data and pilot interferences added constructively at the MRC output and hence their variance scales with $M$, in the same way as the signal power does, and 2) \emph{noncoherent noise} (remaining terms) resulting from all cross-correlation terms (data, pilot, AWGN and QN) which add destructively at MRC output with a variance which doesn't scale with $M$ and hence vanishes as $M \to \infty$. 

Clearly, the coherent noise in both systems is the same which is given by $\bar {\alpha} \rho ^2 K T M$. As a result, we have the following conclusion on the asymptotic performance comparison between QSP and UQSP.
\begin{corollary}
	\label{cor:R_lim_M_singlecell}
	The asymptotic limit of data rates for QSP and UQSP when $M\to \infty$ is
	\begin{equation}
		\label{eq:R_lim_M_singlecell}
		R_{\text{QSP}}^{M \to \infty} =R_{\text{UQSP}}^{M \to \infty}= \log \left ( 1 +\frac{\alpha T}{K} \right ).
	\end{equation}
\end{corollary}
\begin{proof}
	The result follows from taking the limits of both~\eqref{eq: MRC_Rate_QSP} and~\eqref{eq: MRC_Rate_UQSP} when $M\to \infty$ while  all other parameters are fixed.
\end{proof}

The implication of Corollary~\ref{cor:R_lim_M_singlecell} is that, under a quantized channel with superimposed pilots, the QN is averaged out asymptotically. This can be explained as follows. With SP, the channel estimate is highly correlated with received signal since the channel estimate is not independent of data (i.e., see Remark~\ref{rem:CE_data_depednecy}). This gives rise to a coherent-noise term (mainly due to data interference from concurrent users) as discussed previously. On the other hand, there is a noncoherent noise which includes the effect of QN. Whenever the coherent noise is dominated by the noncoherent noise, there's a gap between QSP and UQSP. This occurs when $M$ is not sufficiently large. However, as $M$ increases, the coherent noise starts to dominate the effect of noncoherent noise, thus the gap between QSP and UQSP becomes smaller and smaller until it vanishes asymptotically.

\section{ One-bit quantized multicell massive MIMO}
\label{sec:sec_5}
In this section, we extend the previous results of the single-cell case to a multicell case with a network comprised of $L$ cells and $K$ single-antenna users per cell.

\subsection{Signal model}
For ease of analysis, we assume that the coherence time of the channel can accommodate $KL$ mutually orthogonal pilot sequences, i.e., no pilot contamination is assumed under SP~\cite{7865983}. We also assume that each $K$ users in each cell use power control based on statistical channel inverse as with single-cell scenario.
In the following, we shall use the same notation as before, while introducing indexes of cells. The channel between user $k$ in cell $j$ and the $m$-th antenna of BS $l$ is defined as $\beta_{l j k} h_{l m j k}$. The pilot and data symbols associated with user $k$ in cell $j$, sent during time $t$, are denoted by $c_{j k}[t]$ and $s_{j k}[t]$, respectively. The transmit power of user $k$ in cell $j$ is denoted by $\rho_{j k}$ and the power fractions allocated to pilot and data symbols are denoted by $\alpha_{j k}$ and $\bar \alpha_{j k}$, respectively, such that $\alpha_{j k} + \bar \alpha_{j k}=1$.

Thus the received complex-baseband signal at the $m$th antenna of BS $l$ during time $t$ can be written as
\begin{equation}
	\label{eq:Multicell_BB}
	y_{l m}[t] = \sum_{j=0}^{L-1} \sum_{k=1}^{K}   h_{l m j k}  (\sqrt{\eta_{ljk}} c_{j k}[t] +\sqrt{\bar \eta_{l j k}} s_{j k}[t])+  w_{l m}[t],
\end{equation}
where $\eta_{ljk} = \alpha_{jk}\rho_{jk} \beta_{ljk}$ and $\bar \eta_{l j k} = \bar \alpha_{j k}\rho_{jk} \beta_{ljk}$.  According to the power-control policy, $\rho_{jk}$ is chosen to be $\rho_{jk}= {\beta_{jjk}^{-1}} \rho$, where $\rho$ is some fixed power requirement. As indicated previously, with i.i.d. Rayleigh fading, the average power received at the $j$-th BS from each of its $K$ users is the same, given by $M\rho$.

Let $\theta_{ljk} \triangleq {\beta_{ljk}}/{\beta_{jjk}}$ be the ratio between cross (w.r.t. BS $l$) and direct (w.r.t. to BS $j$) large-scale fading coefficients of user $k$ and, as a result, $\eta_{l j k}$ and $\bar \eta_{l j k}$ reduce to
\begin{equation}
	\eta_{ljk} = \alpha \rho \theta_{ljk}, \qquad \bar \eta_{ljk}= \bar \alpha \rho \theta_{ljk}.
\end{equation}
Since the results are averaged over many realizations (rather than a single realization) of small-scale and large-scale fading coefficients, the statistic of $\{\theta_{ljk}\}$ shall be independent of cell index. This is true due to the random locations of users within each cell which make this assumption reasonable and hence there is no reason to believe they are different for different cell indexes (see~\cite{7438738}). Therefore, without loss of generality, in the sequel, we only focus on BS $0$ as the target BS, i.e., $l=0$.

From~\eqref{eq:Multicell_BB}, the variance of received signal is given by
\begin{align}
	\sigma_{y_{0}}^2  = \kappa_0 \rho+1,
\end{align}
where $\kappa_0 $ is defined by
\begin{equation}
	\label{eq:kappa0}
	\kappa_0 \triangleq  K + \sum_{j=1}^{L-1} \sum_{k=1}^{K}  \theta_{0jk} .
\end{equation}
In addition to~\eqref{eq:kappa0}, we also need the following definition:
\begin{equation}
	\label{eq:kappa1}
	\kappa_1 \triangleq K +  \sum_{j=1}^{L-1} \sum_{k=1}^{K}  \theta_{0jk}^2.
\end{equation}
Note that $\kappa_0$ and $\kappa_1$ are random variables (i.e., depending on users' distances from the target BS), thus we define the following statistics (expected values), which will be used later, as follows:
\begin{equation}
	\label{eq:network_geometry_statistics}
	\zeta_1 \triangleq \E\{\kappa_0\},\quad \zeta_2 \triangleq \E\{\kappa_0^2\}, \quad \zeta_3 \triangleq \E\{\kappa_1\}.
\end{equation}

By the Bussgang decomposition as discussed in Sec.~\ref{sec:sec_3} (also see Fig.~\ref{fig:one-bit system}(b)), the 1-bit quantized version of~\eqref{eq:Multicell_BB} is
\begin{equation}
	\label{eq:multicell_QBB}
	r_{0 m}[t] =  \sqrt{\gamma^{\prime} } y_{0 m}[t]  + z_{0 m}[t],
\end{equation}
where $z_{0 m}[t]$ is the QN and ${\gamma^{\prime} }$ is a scaling factor given by
\begin{equation}
	\label{eq:gamma_muliplecell}
	\gamma^{\prime} = \frac{2}{\pi \sigma_{y_{0}}^2}.
\end{equation}

Let $\mf h_{0 j k}= (h_{0 1 j k}, h_{0 2 j k},\cdots,h_{0 M j k})^{\TT}$ be the vector of small-scale fading gains between user $k$ in cell $j$  and all antennas of BS $0$. Denote by $\mf y_{0}[t]=(y_{0 1}[t],\cdots,y_{0 M}[t])^{\TT}$ the vector of all unquantized signals received at all antennas of BS $0$ during time $t$, $\mf r_{0}[t] = (r_{0 1} [t] ,\cdots,r_{0 M}[t])^{\TT}$ its corresponding quantized version and $\mf z_{0}[t] = (z_{0 1} [t] ,\cdots,z_{0 M}[t])^{\TT}$ \& $\mf w_{0}[t] = (w_{0 1} [t] ,\cdots, w_{0 M}[t])^{\TT}$  the associated QN and AWGN vectors, respectively. We also define $\mf s_{j k}=(s_{jk}[1],\cdots,s_{jk}[T])^{\TT}$ and $\mf c_{j k}=(c_{jk}[1],\cdots,c_{jk}[T])^{\TT}$ as the data  and pilot sequences of user $k$ in the $j$-th cell, respectively. Moreover, we denote by $\mf D_0 = {\text{diag}} \left(\theta_{001}^{\frac{1}{2}}, \cdots,\theta_{0 L-1 K}^{\frac{1}{2}} \right)\in \mathbb{R}^{K L\times KL}$ the diagonal matrix accounting for large-scale fading coefficients, $\mf H_0 = \left(\mf h_{00 1},\cdots,\mf h_{0 L-1 K} \right)\in \mathbb{C}^{M\times KL}$ the composite small-scale fading matrix, $\mf Y_0=\left(\mf y_{0}[t],\cdots, \mf y_{0}[T]\right)\in \mathbb{C}^{M\times T}$ the composite matrix of all received unquantized signals at all antennas during one coherence time, $\mf R_0=\left(\mf r_{0}[t],\cdots, \mf r_{0}[T]\right)\in \mathcal{A}^{M\times T}$ its quantized version, $\mf C = \left(\mf c_{0 1}, \cdots, \mf c_{L-1 K}\right)^{\TT} \in \mathbb{C}^{KL\times T}$ the composite pilot matrix, $\mf S = \left(\mf s_{0 1},\cdots, \mf s_{L-1 K}\right)^{\TT}\in \mathbb{C}^{KL\times T}$ the composite matrix of all data symbols and $\mf W_0=\left(\mf w_{0}[t],\cdots, \mf w_{0}[T]\right)$ \& $\mf Z_0=\left(\mf z_{0}[t],\cdots, \mf z_{0}[T]\right) \in \mathbb{C}^{M\times T}$ are matrices of AWGN and QN, respectively.

From the above discussion, $\mf Y_0 $ can be written as
\begin{equation}
	\mf Y_0 = \sqrt{\alpha \rho } \mf H_0  \mf D_0 \mf C + \sqrt{\bar \alpha \rho } \mf H_0  \mf D_0 \mf S  + \mf W_0,
\end{equation}
and its quantized version as
\begin{align}
	\label{eq:QBB_multiplecell} 
	\mf R_0 = \sqrt{ \gamma^{\prime} } \mf Y_0 + \mf Z_0.
\end{align}

\subsection{Channel estimation}
Let $\hat{ \mf h}_{0 0 k}$ be the channel estimate of user $k$ in the $0$-th cell. Then by the standard results on LMMSE, we can readily show that $\hat{\mf h}_{0 0 k}$ is given by
\begin{equation}
	\hat{\mf h}_{0 0 k} = {\xi^{\prime}} \mf R_0 \mf c_{ 0 k}^{\ast}.
\end{equation}
where
\begin{equation}\label{eq:xi_mutiplecell}
	{\xi^{\prime}} \triangleq \frac{\sqrt{\alpha \rho \gamma^{\prime} }}{ \alpha \rho \gamma^{\prime}   T + \bar \alpha \gamma^{\prime} \rho \kappa_0 + \gamma^{\prime}  + \sigma_z^2 }.
\end{equation}
As indicated previously, we are interested in averaging the results w.r.t. large-scale fading. The following lemma gives an upper bound on the average of variance of estimation error ${\sigma_{\tilde{h}}^{\prime 2}}$.
\begin{lemma}
	\label{lem:MSE}
	When the QN is i.i.d., the expected value (w.r.t. large-scale fading) of ${\sigma_{\tilde{h}}^{\prime 2}}$ is upper bounded by
	\begin{equation}\label{MSE_LB}
		\E \{{\sigma_{\tilde{h}}^{\prime 2}}\} \le 1- \frac{\alpha \rho}{ \alpha \rho  + \frac{(\bar \alpha   + \sigma_z^2 \pi /2)\rho \zeta_1+ \sigma_z^2 \pi /2 +1}{T} }.
	\end{equation}
\end{lemma}
\begin{proof}
	From the orthogonality principle of LMMSE, the variance of estimation error per single realization of $k_0$ is ${\sigma_{\tilde{h}}^{\prime 2}} = 1-T\sqrt{\alpha \rho \gamma^{\prime} } {\xi^{\prime}}$. Note that ${\sigma_{\tilde{h}}^{\prime 2}}$ is a concave function of $\kappa_0$. Thus, taking the expectation of ${\sigma_{\tilde{h}}^{\prime 2}}$ w.r.t. $\kappa_0$ and making use of Jensen's inequality,~\eqref{MSE_LB} follows.
\end{proof}
As discussed previously (see discussion of Lemma~\ref{lemma:QSPS}), the effect of QN, AWGN and data interference is scaled down by a factor of $T$; $\E \{{\sigma_{\tilde{h}}^{\prime 2}}\} \to 0$ as $T\to \infty$. Further, there is a saturation effect as $\rho\to \infty$ (while $T$ is fixed) as the effect of data interference and QN persists to exist, suggesting the saturation of data rate as will be seen in Sec.~\ref{sub:Asymptotic analysis mutliple cell}. Increasing pilot power leads always to an improvement of channel estimate quality while data rate may decrease, however, the relation between channel estimate quality and the achievable data rate is not monotonic as discussed next. This implies that there is a trade-off between channel estimate quality and data rate dictated by power allocation strategy.

It should be noted that~\eqref{MSE_LB} serves as an approximate upper bound (with some artifact effect) as there will still be some correlation between QN components. It turns out that the bound~\eqref{MSE_LB} is nearly tight in many cases (especially in the low-SNR and large-$K$ regimes) as demonstrated by numerical results, i.e., see Fig.~\ref{fig:MSE}.
\subsection{Analysis of achievable rates}
Now we present an approximate lower bound on the achievable rate for QSP under multicellular case with MRC implemented at BSs and no pilot removal after estimating the channel is performed. We next extend the result for the unquantized system. 

Without loss of generality we can assume an arbitrary user $k$ in cell $0$ as our target user. Thus, the scaled MRC output during time $t$ is 
\begin{equation} \label{eq:multicell_MRC_out}
	\hat{s}_{0 k} [t]\triangleq \frac{1}{M}\hat{\mf h}_{0 0 k}^{\HH} {\mf r}_{0} [t]=\frac{\xi^{\prime}}{M} \mf c_{0 k}^{\TT} \mf R_0^{\HH} {\mf r}_{0}[t].
\end{equation}
Based on~\eqref{eq:multicell_MRC_out}, a closed-form approximation on the achievable rate is given in the following theorem. 
\begin{thm}[approximate lower bound]
	\label{thm: MRC_Rate_QSP_multiplecell}
	Consider 1-bit QSP multicell massive MIMO where QN is assumed to be i.i.d., and LMMSE channel estimator and MRC receiver are employed at the BS. If power control based on statistical channel inverse is applied, then a lower bound on the achievable rate in uplink is approximated by
	\begin{equation} \label{eq: MRC_Rate_QSP_multiplecell}
		R_{\text{QSP}}^{\prime} \approx \log \left( 1+ \Upsilon_{\text{QSP}}^{\prime} \right) \quad (\text{bits/s/Hz}),
	\end{equation}
	where $\Upsilon_{\text{QSP}}^{\prime}$ is given by~\eqref{eq:SNR_multiplecell_QSP}, shown at the top of page~\pageref{eq:SNR_multiplecell_QSP}.
\end{thm}
\begin{proof}
	See Appendix~\ref{app: proof_thm2}.
\end{proof}

\begin{figure*}[!btp]
	\normalsize
	\begin{IEEEeqnarray}{rCl}  \nonumber
		\Upsilon_{\text{QSP}}^{\prime}&=&\Big({\alpha  \bar{\alpha}   \rho ^2 T^2 M}\Big) \Big / \Big(  \bar {\alpha} \rho ^2 M T  \zeta_3 + \frac{1}{4}\left( \alpha ^2 \rho ^2-8 \alpha  \rho ^2+4 \rho ^2\right) \zeta _3 + \frac{1}{4} \left(\pi ^2 \rho ^2 T -\pi ^2 \rho ^2 - 4 \alpha  \rho ^2 T \right)\zeta_2\\
		&&+ \frac{1}{4} \left(4 \alpha  \rho ^2 T^2 -2 \pi ^2 \rho  - 4 \alpha  \rho  T+2 \pi ^2 \rho  T\right)\zeta_1+ \frac{1}{4} \left(4 \alpha  \rho  T^2-8 \alpha ^2 \rho ^2 T+ 8 \alpha  \rho ^2 T+\pi ^2 T-\pi^2\right)  \Big)
		\label{eq:SNR_multiplecell_QSP}
	\end{IEEEeqnarray}
	\hrulefill
	\vspace*{0pt}
\end{figure*}

A corollary of Theorem~\ref{thm: MRC_Rate_QSP_multiplecell} which maximizes $R_{\text{QSP}}^{\prime}$ w.r.t. $\alpha$ is the following.
\begin{corollary} 
	\label{cor:optimal_alpha_QSP_multicell}
	The optimal power fraction $\alpha^{\ast} \in (0,1)$ which maximizes~\eqref{eq: MRC_Rate_QSP_multiplecell} is given by one of the two roots:
	\begin{equation}  \label{eq:optimal_alpha_QSPMultiplecells}
		\alpha^{\ast}=\frac{-\delta \pm \sqrt{\delta  \delta^{\prime}}}{4 \rho  \left(\zeta_1 T (\rho  T-1)-\rho  T (\zeta_2+\zeta_3 M)-\zeta_3 \rho +T^2\right)}
	\end{equation}
	where $\delta = \pi ^2 (T-1) \left(2 \zeta_1 \rho +\zeta_2 \rho ^2+1\right)+4 \zeta_3 M \rho ^2 T-4 \zeta_3 \rho ^2$ and
	$\delta^{\prime} =\pi ^2 (T-1) \left(2 \zeta_1 \rho +\zeta_2 \rho ^2+1\right)+4 \rho  T (\zeta_1 \rho  T-\zeta_1-\zeta_2 \rho +T)$.
\end{corollary}
\begin{proof}
	The proof follows by following same lines of the proof of Corollary~\ref{cor:optimal_alpha_QSP}.
\end{proof}

For the sake of comparison, we introduce the following conclusion on the achievable rate for UQSP.
\begin{corollary}[lower bound]
	\label{cor:MRC_Rate_UQSP_multiplecell}
	Under UQSP, a lower bound on the achievable rate in uplink is
	\begin{equation}
		\label{eq: MRC_Rate_UQSP_multiplecell}
		R_{\text{UQSP}}^{\prime} = \log (1+ \Upsilon_{\text{UQSP}}^{\prime})
	\end{equation}
	where $\Upsilon_{\text{UQSP}}^{\prime}$ is defined as
	\begin{IEEEeqnarray}{rCl}
		\nonumber
		&&\Upsilon_{\text{UQSP}}^{\prime} =\Big({\alpha  \bar \alpha  \rho ^2 T^2 M }\Big)  \Big/ \Big(\bar {\alpha}\rho^2 M T \zeta_3 + \bar {\alpha}^2 \rho ^2\zeta_3  + \bar \alpha   \rho ^2 T \zeta_2 \\ \nonumber
		&&+\> \left(\alpha  \rho ^2 T^2  + \alpha  \rho  T+2 \bar \alpha   \rho  T \right)\zeta_1  + \alpha  \rho  T^2 + 2 \alpha  \bar \alpha   \rho ^2 T + T \Big).\\
		\label{eq:SNR_multiplecell_UQSP}
	\end{IEEEeqnarray}
	Furthermore, the optimal power fraction $\alpha^{\ast} \in (0,1)$ which maximizes~\eqref{eq: MRC_Rate_UQSP_multiplecell} is given by one of the following two roots:
	\begin{equation}
		\label{eq:optimal_alpha_UQSP_multicell}
		\alpha^{\ast} = \frac{-\delta \pm \sqrt{  T (\rho \zeta_1 +1) (\rho  T+1) \delta} }{ \rho T (\rho  T-1)\zeta_1 -  \rho ^2 T \zeta_2 - \rho  (M \rho  T+\rho )\zeta_3 + \rho  T^2},
	\end{equation}
	where $\delta = 2 \rho  T \zeta_1 +  \rho ^2 T \zeta_2 +    (\rho ^2+M \rho^2 T)\zeta_3 + T$.
\end{corollary}
\begin{proof}
	From Appendix~\ref{app: proof_thm2}, setting $\sigma_z=0, \gamma^{\prime}=1$ and $f^{\prime}(t,t)=(M+1)\sigma _{y_0}^4/M$ in~\eqref{eq:normalized_noise_multiplecell} and substituting the result in~\eqref{eq:avg_rate_multicell},~\eqref{eq: MRC_Rate_UQSP_multiplecell} follows. The second part follows the same lines of proof as of Corollary~\ref{cor:optimal_alpha_QSP}.
\end{proof}

\subsection{Asymptotic analysis}
\label{sub:Asymptotic analysis mutliple cell}
Having obtained the achievable rates for the quantized and unquantized systems, it is interesting to know the gap between them. To that end, we use the asymptotic analysis as we have done previously for the single-cell case. With the single-cell case, we have classified the noise at the MRC output into coherent and noncoherent noises.  Likewise, we can see that the denominators (total noise power at MRC output) of $\Upsilon_{\text{QSP}}^{\prime}$~\eqref{eq:SNR_multiplecell_QSP}  and $\Upsilon_{\text{UQSP}}^{\prime}$~\eqref{eq:SNR_multiplecell_UQSP} are partitioned into \textit{coherent} and \textit{noncoherent} noises, where the latter in $\Upsilon_{\text{QSP}}^{\prime}$  includes the overall effect of QN which does not scale with $M$. From~\eqref{eq:SNR_multiplecell_QSP} and~\eqref{eq:SNR_multiplecell_UQSP}, the first term of both denominators, which is the same, corresponds to the coherent noise given by $\bar {\alpha} \rho ^2 M T  \zeta_3$. The following corollary is immediate.
\begin{corollary}
	\label{cor:R_lim_M}
	The asymptotic limit of data rates for QSP and UQSP when $M\to \infty$ is
	\begin{equation}
		\label{eq:R_lim_M_multicell}
		R_{\text{QSP}}^{\prime M \to \infty} =R_{\text{UQSP}}^{\prime M \to \infty}= \log \left ( 1 +\frac{\alpha T}{\zeta_3} \right ).
	\end{equation}
\end{corollary}

Corollary~\ref{cor:R_lim_M} immediately implies that in quantized multicell massive MIMO scenarios with superimposed pilots, QN can be averaged out asymptotically in $M$. Note that Corollary~\ref{cor:R_lim_M} is similar to Corollary~\ref{cor:R_lim_M_singlecell} except the factor $\zeta_3$ which accounts for the inter-cell interference. Interpreting Corollary~\ref{cor:R_lim_M} follows the same logic made under Corollary~\ref{cor:R_lim_M_singlecell} in Sec.~\ref{subsec:Asymptotic analysis_singlecell}.

We remark that the achievable rate $R_{\text{UQSP}}^{\prime}$ in Corollary~\ref{cor:MRC_Rate_UQSP_multiplecell} and the asymptotic rate in Corollary~\ref{cor:R_lim_M} for the unquantized system can be derived as special cases of the results in~\cite{2017arXiv170907722V}  under our network settings as discussed previously.

Moreover, we have the following result:
\begin{corollary}
	\label{cor:R_lim_SNR}
	The asymptotic limits of data rates for QSP and UQSP when $\rho\to \infty$ are, respectively, given by
	\begin{align}
		R_{\text{QSP}}^{\prime \rho \to \infty} &= \log \left ( 1+\Upsilon_{\text{QSP}}^{\prime \rho\to \infty} \right), \label{eq:R_lim_SNR_QSP} \\
		R_{\text{UQSP}}^{\prime \rho \to \infty} &= \log\left ( 1+\Upsilon_{\text{UQSP}}^{\prime \rho \to \infty} \right), \label{eq:R_lim_SNR_UQSP}
	\end{align}
	where $\Upsilon_{\text{QSP}}^{\prime \rho\to \infty}$ and $\Upsilon_{\text{UQSP}}^{ \prime \rho \to \infty}$ are defined by
	\begin{subequations}
		\begin{align}\nonumber
			\Upsilon_{\text{QSP}}^{\prime \rho\to \infty}   &= \left({\bar{\alpha } \alpha  M T^2}\right) \big / \big(  \alpha  T^2 \zeta _1 - \frac{1}{4}\left(4 \alpha  T - \pi ^2 T + \pi ^2\right) \zeta _2 \\
			&  \bar{\alpha} \left(\bar{\alpha } +   M T\right) \zeta _3 + 2 \bar{\alpha } \alpha  T  \big), \\
			\Upsilon_{\text{UQSP}}^{\prime \rho \to \infty} &= \frac{\alpha \bar{\alpha } M T^2}{ \alpha T^2 \zeta _1 + \bar{\alpha } T \zeta _2+  \bar{\alpha }(\bar{\alpha } +M T) \zeta _3 + 2 \alpha \bar{\alpha }   T }.
		\end{align}
	\end{subequations}
\end{corollary}
\begin{proof}
	The results follows from taking the limits of~\eqref{eq: MRC_Rate_QSP_multiplecell} and~\eqref{eq: MRC_Rate_UQSP_multiplecell} when $\rho \to \infty$.
\end{proof}
Corollary~\ref{cor:R_lim_SNR} is a natural consequence of the saturation effect of LMMSE due to QN and data interference which scale with $\rho$ (see the discussion of Lemma~\ref{lem:MSE}). Equations~\eqref{eq:R_lim_M_multicell},~\eqref{eq:R_lim_SNR_QSP} and~\eqref{eq:R_lim_SNR_UQSP} can be further maximized w.r.t. $\alpha$ as explained in the proof of Corollary~\ref{cor:optimal_alpha_QSP}. For instance, by inspecting Corollary~\ref{cor:R_lim_M}, it is interesting to note that when $M \to \infty$, the optimal policy to maximize the data rate is to allocate all power to pilot, i.e., $\alpha \to 1$.

We remark here that the phenomenon of the continuous increase of pilot power with the increase of the number of BS antennas was also observed in~\cite{7865983,2017arXiv170907722V} under the unquantized MIMO channel with SP scheme.

\section{Numerical results}
\label{sec:sec_6}
In this section, we present some numerical results to compare the analytical bounds in Theorem~\ref{thm: MRC_Rate_QSP_multiplecell}, Corollaries~\ref{cor:optimal_alpha_QSP_multicell}-\ref{cor:R_lim_SNR} and Lemma~\ref{lem:MSE} against the results obtained by Monte Carlo (MC) simulation. Furthermore, we include the simulation result for QTP, where we assume $\tau=K$ and no optimization w.r.t. power or training length is performed (non-optimized QTP). With QTP, we assume each user is assigned a pilot sequence randomly from a fixed set of pilot sequences and the pilot sequences of each $K$ users within the same cell are mutually orthogonal. In all MC simulations, the rates are maximized (w.r.t. $\alpha$) per single realization of large-scale fading and thus the final rate is taken to be the average over many such realizations. Moreover, we include the simulation results when PR is implemented after estimating the channel as described by~\eqref{eq:PR_model}.

In our model, we consider a hexagonal network of one tier of BSs, i.e., $L=7$. Each cell has a radius $r_{\text{c}} = 1.8$ Km with a forbidden region of radius $r_{\text{f}} = 0.1$ Km in which no user exists. The channel is assumed flat-block Rayleigh fading with bandwidth $B_{\text {w}}$ = $200$ KHz and each block is transmitted within $1$ ms (i.e., one sub-frame in LTE). This translates into a coherence time $T = 200$ (in symbol intervals). All users are distributed uniformly and randomly within each cell. The large-scale fading coefficient $\beta_{0jk}$ is defined by $\beta_{0jk} = \omega^{-1} d_{0jk}^{-\zeta}$~\cite{7438738}~\cite{2017arXiv170907722V}, where  $d_{0 j k}$ is the distance (in Km) from the $k$-th user in cell $j$ to BS $0$, $\zeta = 3.8$ is the path-loss exponent and $\omega$ is the path-loss at a reference distance of $1$ Km, which also accounts for distance-independent propagation losses such as wall penetration. The statistics $\zeta_1, \zeta_2$ and $\zeta_3$ defined in~\eqref{eq:network_geometry_statistics} which are required for computing the analytical results are obtained by MC simulation, where they converge to fixed values depending on $K$ and the geometry of the network, such as the radii of cell and forbidden region. According to our cell settings, we observe that including more tiers (say, second tier, $L=19$) changes the computed statistics slightly. For convenience, the system parameters are summarized in Table~\ref{tab:simulationparameters}.

\begin{table}
	\caption{Summary of simulation parameters.}
	\label{tab:simulationparameters}
	\setlength{\tabcolsep}{3pt}
	\centering
	\begin{tabular}{l|p{115pt}}
		\hline
		Parameter &
		Value/Description \\
		\hline
		Cell layout & One-tier hexagonal, $L=7$ \\
		System bandwidth ($B_{\text {w}})$ & $200$ KHz \\
		Coherence time ($T$) & 200 symbol intervals (1ms)\\
		Cell radius ($r_{\text{c}}$) & $1.8$ Km \\
		Forbidden region radius ($r_{\text{f}}$)& $0.1$ Km \\
		Path-loss exponent ($\zeta$) & 3.8 \\
		$\zeta_1=\E\{k_0\}$ & $\approx  1.4116 K$ \\
		$\zeta_3=\E\{k_1\}$ & $\approx  1.1656K$ \\
		$\zeta_2=\E\{k_0^2\}$ & $\approx 50.53, 288.6, 450.66,1248.94$ when $K=5,12,15,25$\\ \hline
	\end{tabular}
\end{table}

\begin{table}[htp]
	\caption{The average optimal fraction ${\alpha}^{\ast}$ of total power for different number of BS antennas, $K = 12, T = 200, \text{SNR} = -10 \text{ dB}$.}
	\label{tab:table1}
	\centering
	\begin{tabular}{l|c |c|c|c}
		\hline
		No. of BS antennas $M$&
		50&
		200&
		600&
		1000\\
		\hline
		MC QSP                      &0.34&         0.41&        0.50&     0.59              \\ \hline
		MC QSP PR                    &0.38&        0.43&        0.53&    0.59               \\ \hline
		Anal. QSP~\eqref{eq:optimal_alpha_QSPMultiplecells}       &0.38    &        0.45&        0.55     &0.61 \\ \hline
		
		MC UQSP&0.33&         0.44&        0.57&     0.65\\ \hline
		MC UQSP PR &0.40&        0.49&        0.6&    0.64\\ \hline
		Anal. UQSP~\eqref{eq:optimal_alpha_UQSP_multicell} &0.33&        0.44&        0.56&0.62\\ \hline
	\end{tabular}
\end{table}

In Fig.~\ref{fig:MSE} the empirical variance of channel estimation error is compared with the bound in Lemma~\ref{lem:MSE}. From Fig.~\ref{fig:MSE}, it is clear that the analytical bound serves as a good approximation of the average MSE of the channel estimate for all SNR values, especially when SNR is low. This result is a consequence of the law of large numbers which is explained as follows. From the definition~\eqref{eq:network_geometry_statistics}, we have $\zeta_1 \triangleq \E\{\kappa_0\}=K+\E \{\sum_{j=1}^{L-1} \sum_{k=1}^{K}  \theta_{0jk}\}=K+\sum_{j=1}^{L-1} \sum_{k=1}^{K} \E\{\theta_{0jk}\}=K+K(L-1) \bar \theta$, where the last equality follows from the fact that, the positions (and hence distances) of all users in the network are independent random variables. Further, within the same network tier, the positions of all users (w.r.t. BS 0) are i.i.d. random variables as users are distributed uniformly and randomly within each cell. We next notice that $\bar \theta = \lim_{K\to \infty} \frac{1}{K (L-1)} \sum_{j=1}^{L-1} \sum_{k=1}^{K}  \theta_{0jk}$, implying that $\zeta_1$ approaches a single-realization of the random variable $\kappa_0$ as $K$ gets large, thus the Jensen's inequality in Lemma~\ref{lem:MSE} becomes tight for asymptotically large $K$. Considering more than one tier, the users are partitioned into different groups (each group belongs to one network tier), where the positions of all users within each group are i.i.d. random variables, thus the analysis follows the same logic as before on each group of users.  Finally, we can see that the MSE performance improves as $T$ increases which is an intuitive result as inspected from Lemma~\ref{lem:MSE}.

\begin{figure}
  \centering
   \includegraphics[width=0.4\textwidth]{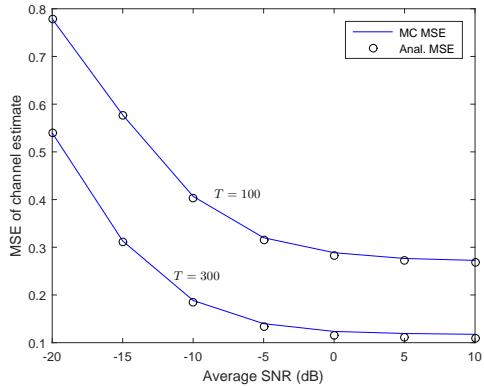}\\
  \caption{Average SNR versus MSE of channel estimate for a fixed power allocation $\alpha =0.5$ in QSP, $L=7, K = 12, M = 100$.}
  \label{fig:MSE}
\end{figure}

Figure~\ref{fig:Rate_SNR} illustrates the average SNR versus the achievable rates for the QSP, UQSP, and QTP. The analytical asymptotes are obtained from Corollary~\ref{cor:R_lim_SNR}. With QSP, we observe the analytical bound~\eqref{eq: MRC_Rate_QSP_multiplecell} works as a good approximation for the rate under QSP, especially when PR is used. Further, the analytical approximation is more accurate in the low-SNR regime and overestimates the achievable rate slightly in other SNR regions. The latter should come as no surprise due to the approximation made in the derivation of the approximate bound~\eqref{eq: MRC_Rate_QSP_multiplecell} of Theorem~\ref{thm: MRC_Rate_QSP_multiplecell}, where the variance of effective noise at the MRC output is underestimated, mainly due to the ideal i.i.d. assumption on QN (see Appendices~\ref{app: proof_thm2} and~\ref{app: proof_thm1}). For the unquantized system with the no-PR assumption, the lower bound~\eqref{eq: MRC_Rate_UQSP_multiplecell} and MC result are almost the same\footnote{Note that when deriving~\eqref{eq: MRC_Rate_UQSP_multiplecell}, no approximations are used as those made under QSP. Thus any gap between the simulation and analytical results is due to Jensen's inequality, see the proof of Corollary~\ref{cor:MRC_Rate_UQSP_multiplecell}.}. Also, it can be seen that QSP outperforms the non-optimized QTP in all simulated cases.

As shown in Fig.~\ref{fig:Rate_SNR}, all rates in all systems saturate and converge as SNR grows large. Undoubtedly, this is an expected result as discussed previously (see discussion of Lemma~\ref{lem:MSE}). With a moderate $M=100$, we observe that the gap between QSP and  UQSP with and without PR assumption is almost the same; ranges from 0.144 bps/Hz at $-20$ dB (very low SNR) to $0.52$ bps/Hz at 10 dB (high SNR), which suggests that with 1-bit QSP, the loss in information is not significant compared with the infinite-resolution case. Thus, when the number of quantization bits increases (say, 2 or 3 bits), this gap becomes less pronounced.
\begin{figure}
  \centering
   \includegraphics[width=0.4\textwidth]{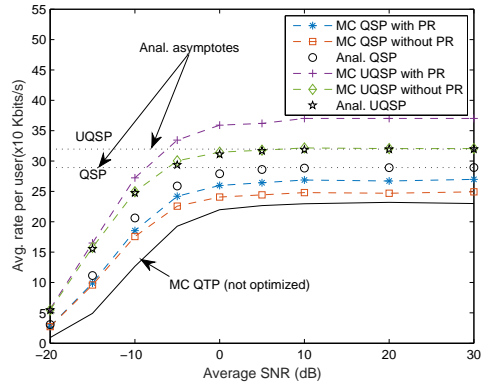}\\
  \caption{The achievable rate vs. SNR for QSP and UQSP, $L=7, K=12, M=100, T=200$.}
  \label{fig:Rate_SNR}
\end{figure}

Figure~\ref{fig:Rate_K} shows the impact of the number of users $K$ on the achievable rate for the quantized system\footnote{Notice that $K$ enters the effective SNR equation~\eqref{eq:SNR_multiplecell_QSP} through the statistics $\zeta_1,\zeta_2$ and $\zeta_3$}. The results indicate that as $K$ increases, the per-user data rate decreases, which is a natural consequence of the increase of data interference of all users. As expected, we observe that when $K$ increases or SNR is low, our analytical approximation becomes more accurate, since by increasing $K$ or working in the low-SNR regime, the accuracy of the i.i.d. assumption on QN will increase (see Assumption 3). In massive MIMO, highly loaded networks with tens of users served simultaneously by each BS is expected~\cite{5595728}, thereby making our analytical bound~\eqref{eq: MRC_Rate_QSP_multiplecell} a good approximation of the performance. Finally, we observe that in almost all cases, QSP outperforms QTP. However, the gap in performance becomes less pronounced as $K$ increases.


\begin{figure}
\centering
   \includegraphics[width=0.4\textwidth]{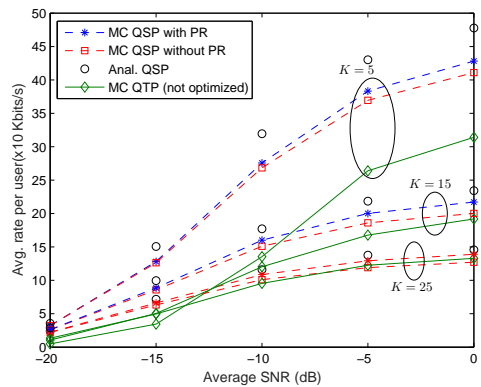}\\
  \caption{The impact of number of users $K$ on the achievable rate for QSP, $L=7, M = 100, T = 200$.}
  \label{fig:Rate_K}
\end{figure}

Figure~\ref{fig:Rate_M} depicts the impact of increasing BS antennas $M$ on the achievable data rate. We also show in Table~\ref{tab:table1} the average power fraction $\alpha$ (of pilot) obtained by simulation for both the quantized and unquantized systems when $M=50,200,600,1000$. As shown in Table~\ref{tab:table1}, the common trend for the optimal value of $\alpha$ is that it increases as $M$ increases. This is consistent with Corollary~\ref{cor:R_lim_M} which predicts such a phenomenon where it was shown that the optimal policy to maximize data rate when $M$ is asymptotically large is to allocate most power to pilot. From Fig.~\ref{fig:Rate_M}, we can see that, for all cases, increasing $M$ gives rise to an increase in data rates. Once again, we observe that QSP outperforms the non-optimized QTP as shown previously.

Obviously, higher data rates can be achieved when the BS employs PR technique using the channel estimate, compared with the no-PR case. However, the gap between the PR and no-PR cases is not significant, especially under the quantized system. Interestingly, the analytical approximation for QSP serves as a good approximation of the achievable data rate, particularly when PR is used at the BS. The reason of this can be explained as follows. In the analytical analysis of QSP, many terms contributing small quantifies to the noise variance are neglected due to the approximations (see, for example, the approximations of $a_{0k}$ in~\eqref{eq:a0k} and $f_2^{\prime}-f_{10}^{\prime}$ in~\eqref{eq:zeta_fn_multicell}), thus the variance of effective noise at MRC output is underestimated. With PR used, the variance of noise is further reduced and hence this reduction of noise variance, in some sense, compensates for the underestimation of noise variance incurred in the analytical bound. This is clear from the closeness between the analytical approximation and simulated QSP with PR. Compared with the unquantized system, in the quantized case more antennas are required to achieve the same data rate, which is expected due to QN. For a small or medium number of BS antennas, we observe that the ratio between the number of antennas for QSP to the number of antennas for UQSP is roughly 2. However, this ratio decreases gradually as $M$ gets larger due to the asymptotic convergence of data rates for both schemes, as will be seen next.

\begin{figure}
  \centering
   \includegraphics[width=0.4\textwidth]{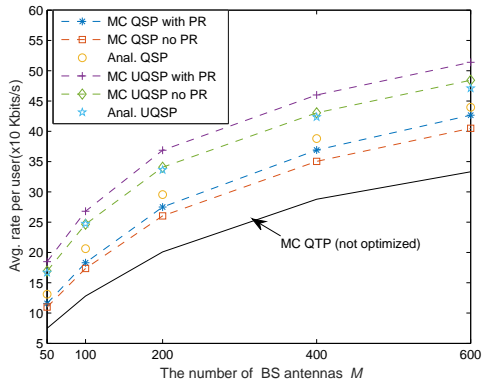}\\
  \caption{The achievable rate vs. no. of BS antennas $M$, $L=7, K = 12, T=200$ and ${\rm{SNR}}=-10\rm{dB}$.}
  \label{fig:Rate_M}
\end{figure}

Figure~\ref{fig:Lim_M} shows the asymptotic behavior of the data rate when $M$ grows large. The analytical asymptote in Corollary~\ref{cor:R_lim_M} is also shown for comparison. As can be seen from Fig.~\ref{fig:Lim_M}, for QSP, UQSP, and QTP systems, the data rates increase with the increase of $M$ and finally converge to fixed values. Specifically, the data rates of QSP and UQSP systems approach the analytical asymptote for a very large number of BS antennas. This implies that, for asymptotically large $M$, quantization incurs no loss of information, when compared with the infinite-resolution counterpart in SP systems. Understanding this interesting phenomenon is discussed in Sec.~\ref{subsec:Asymptotic analysis_singlecell} under Corollary~\ref{cor:R_lim_M_singlecell}.


\begin{figure}
	\centering
	\includegraphics[width=0.4\textwidth]{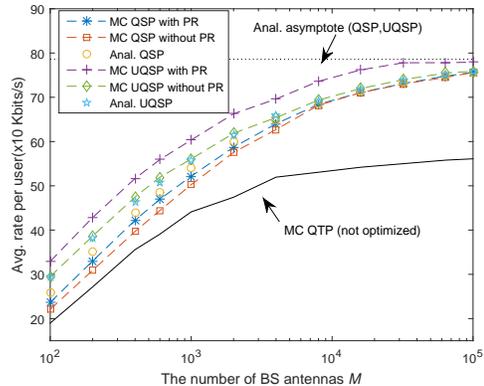}\\
	\caption{The asymptotic behaviour of data rate with increasing the number of BS antennas, $L=7, K = 12, T=200, \text{SNR}=-5 \text{ dB}$.}
	\label{fig:Lim_M}
\end{figure}

\section{Conclusions and Future Work}
\label{sec:sec_7}
In this work, we consider the achievable data rate in the uplink of a 1-bit
quantized massive MIMO system when superimposed pilot
scheme (i.e., QSP) is used. We derive an approximate lower
bound on the achievable rate for QSP with the assumption
of i.i.d. QN. We have also recovered a true lower bound
on its infinite-resolution counterpart (i.e., UQSP) as a special
case of QSP. We have showed that regardless the coarse
quantization and pilot-data superposition, high data rates can
be achieved in practical multicell scenarios. Although there
is a gap between QSP and UQSP due to QN, we show that
the effect of QN diminishes gradually as the number of BS
antennas increases. This is because when SP is used, the
coherent noise at the MRC output dominates the effect of noncoherent
noise which includes the effect of QN. Interestingly,
for asymptotically large $M$, the rates achieved under QSP
and UQSP saturate and converge to the same deterministic
value. Thus we conclude that QSP is interference-limited
rather than QN-limited, asymptotically.

An important result of QSP is that when working in the low-SNR regime or when the number of users is sufficiently
large, pilot elimination after estimating the channel does
not result in a significant increase in information, compared
with the case when no pilot removal is considered. As an
illustrative performance comparison only, we have simulated
the achievable rate for the non-optimized QTP (conventional
training scheme with the smallest number of pilot symbols
and equal power allocations among pilot and data), thereby; it
is demonstrated that it is possible for QSP to outperform QTP
in many cases. These results are far from being conclusive
as rigorous investigations are required when optimization in
QTP is considered alongside the use of more advanced signal
processing techniques such as joint pilot-data processing,
which exceeds the scope of this paper.

\appendix
\begin{appendices}
	
	\section{Proof of Theorem 1}
	\label{app: proof_thm1}
	In this section, we prove Theorem~\ref{thm: MRC_Rate_QSP}. To that end, we need the following lemmas.
	\begin{lemma}\label{lem: xxx0}
		\begin{subequations}
			\begin{align}
				\E \left \{ \norm{\mf y[t]}^2 \right \}&=M\sigma_y^2 \label{sub:y_norm2}, \\
				\E \left \{ \norm{\mf  y[t]}^4 \right \}&= M(M+1)\sigma_y^4 .\label{sub:y_norm4}
			\end{align}
		\end{subequations}
	\end{lemma}
	\begin{proof}
		The above results follow straightforwardly since the entries of $\mf y[t]$ are assumed to be i.i.d. $\mathcal{CN} (0,\sigma_y^2)$.
	\end{proof}
	\begin{lemma}\label{lem: xxx1}
		For any two time instants $n$ and $q$, we have
		\begin{subequations}
			\begin{align}
				\lim_{M\to \infty} \frac{ \mf  z^{\HH}[n]  \mf  y[q]}{M} &= \E \left \{z_1^{\ast} [n] y_1[q]\right\} = 0,{}\forall n,q \label{eq:asymp1} \\
				\lim_{M\to \infty} \frac{ \mf  z[n]^{\HH}  \mf  z[q]}{M} &= \left\{ \begin{array}{l}
					\E \left \{z_1^{\ast} [n] z_1[q]\right \} = 0,{} n \ne q\\
					\E \left \{|z_1[n]|^2\right\} = \sigma_z^2,{}  n = q
				\end{array} \right. \label{eq:asymp3}
			\end{align}
		\end{subequations}
	\end{lemma}
	\begin{proof}
		The above results unfold from the fact that each column of $\mf Y$ and $\mf Z$ consists of i.i.d. entries, each with the same mean and variance. Thus each inner product in Lemma~\ref{lem: xxx1} consists of a sum of $M$ i.i.d. entries. By the virtue of the law of large numbers, as $M \to \infty$, the right-hand sides (RHSs) of~\eqref{eq:asymp1} and ~\eqref{eq:asymp3} follow from~\eqref{qnoise_inp_crosscorrelation} and~\eqref{eq: noise_covariance}, respectively.
	\end{proof}
	\begin{lemma}\label{lem: xxx2}
		For any $n \neq q$ and QN is approximately i.i.d., the following conditional densities satisfy (approximately)
		\begin{align}
			p \left(\mf z[n]|\mf z[q]\right)= p\left(\mf z[n]\right) \label{eq:cond1}\\
			p \left (  z[n] | y[q] \right) = p\left (z[n] \right) \label{eq:cond2}.
		\end{align}
	\end{lemma}
	\begin{proof}
		Equation~\eqref{eq:cond1} follows because $\mf z[n]$ is independent of $\mf z[q]$. Since QN is a function of the unquantized signal and the i.i.d. assumption on QN implies that QN at time $n$ is a function of the unquantized signal received at time $n$ only and independent of any signal received at other time instants, hence~\eqref{eq:cond2} follows.
	\end{proof}
	\begin{lemma} \label{lem: xxx3}
		For any $n \ne q$ we have
		\begin{equation}
			\E \left \{ \mf y[n] | \mf y[q]\right \} = \frac{\alpha \rho \bar {\mf c}^{\HH}[q] \bar {\mf c}[n]}{\sigma_y^2}  \mf y[q]
		\end{equation}
	\end{lemma}
	\begin{proof}
		The result follows from applying standard MMSE~\cite{sayed2011adaptive}. This is because the received unquantized signal vector is jointly Gaussian (Assumption 2), the conditional expectation $\E \left \{ \mf y[n] | \mf y[q]\right \}$ corresponds to the MMSE of $\mf y[n] $ given the observation $\mf y[q] $. 
	\end{proof}
	
	Now, by decomposing the output of MRC~\eqref{eq:process_sig} into two parts; signal and uncorrelated effective noise,~\eqref{eq:process_sig} can be rewritten as
	\begin{align}\label{eq:MRC_decomposition}
		\hat s_{k}[t] &= a_{k} s_{k}[t]+  \epsilon_{k} [t],
	\end{align}
	where $a_{k}$ is a deterministic constant and $\epsilon_{k}[t]$ is the effective non-Gaussian noise which is uncorrelated with $s_{k}[t]$. From the orthogonality principle, we have
	\begin{equation}\label{eq:def_ak}
		a_{k} = \E  \{s_{k}^{\ast}[t] \hat s_{k}[t]\}.
	\end{equation}
	and hence the variance of noise $\epsilon_{k} [t]$ is given by $\sigma_{\epsilon_k}^2 =  \E \{|\hat {s}_{k}[t]|^2\}-| \E \{\hat {s}_{k}[t]\}|^2-|a_{k}[t]|^2$. For the sake of comparison with the unquantized system (UQSP), we define
	\begin{equation}
		\tilde{\sigma}_{\epsilon_k}^2 =  \frac{\sigma_{\epsilon_k}^2}{|a_{k}[t]|^2}
		\label{eq:normalized_var_eff_noise}
	\end{equation}
	as the \textit{normalized} variance of effective noise. Since $s_{k}[t]$ is Gaussian, a lower bound on channel capacity is obtained by replacing $\epsilon_{k} [t]$ by a Gaussian noise with the same variance $\sigma_{\epsilon_k}^2$. Thus the lower bound is given by
	\begin{equation} \label{eq:Stnd_lowerbound}
		R_{\text{single-cell}}^{\text{LB}}=\log\left( 1 + \frac{1}{\tilde{\sigma}_{\epsilon_k}^2}\right).
	\end{equation}
	
	\emph{1) Calculation of $a_k$}: Expanding~\eqref{eq:process_sig} yields
	\begin{IEEEeqnarray}{rCl}\label{eq:skt_expansion_0}\nonumber
		\hat s_k[t] &=& \frac{\xi}{M} \sum_{n=1}^{T} c_{k} [n] \Big( \gamma \mf  y^{\HH}[n] \mf  y[t]+\mf  z^{\HH}[n]  \mf  z[t]\\
		&& \quad +\>   \sqrt{\gamma} \mf  y^{\HH}[n]  \mf  z[t] + \sqrt{\gamma}   \mf  z^{\HH}[n]  \mf  y[t] \Big).
	\end{IEEEeqnarray}
	Using~\eqref{eq:skt_expansion_0} in~\eqref{eq:def_ak} gives
	\begin{align} \label{eq:a_kt_evaluation}  \nonumber
		a_{k} &\approx \frac{\xi}{M} \sum_{n=1}^{T} c_{k} [n]  \gamma \mf  \E \left\{  \mf  y^{\HH}[n]  \mf  y[t] s_{k}^{\ast}[t] \right\}\\
		&= \sqrt{\bar \alpha \alpha} \rho  \gamma \xi T,
	\end{align}
	where the approximation in the first line follows from using Lemma~\ref{lem: xxx1}.
	
	\emph{2) Calculation of $\tilde{\sigma}_{\epsilon_k}^2$}: Using~\eqref{eq:skt_expansion_0}, we can verify that
	\begin{align} \label{eq:Expected_shat} \nonumber
		\E\{\hat s_k[t]\}&= \xi (\alpha \rho\gamma  T + \bar \alpha \rho \gamma   K +  \gamma +\sigma_z^2   )  c_{k}[t]\\
		&=\sqrt{\alpha \rho \gamma } c_{k}[t],
	\end{align}
	where in~\eqref{eq:Expected_shat} we have made use of Lemma~\ref{lemm:ch_QN_Correlation}, the assumption that QN samples are i.i.d. each with zero-mean and variance $\sigma_z^2$ (Assumption 3), and~\eqref{eq: noise_covariance}. From~\eqref{eq:process_sig}, $\E \{|\hat {s}_{k}[t]|^2\}$ can be written as
	\begin{align} \nonumber
		\label{eq:hat s_kt_squared}
		&\E\{|\hat {s}_{k} [t] |^2\} =  \frac{\xi^2}{{M}^2}  \mf {c}_k^{\TT}  \E \left \{\mf  R^{\HH}  \mf  r [t] \mf  r [t]^{\HH} \mf  R \right \}\mf {c}_k^{\ast} \\
		&=\xi^2 f(t,t)  + {\xi^2}  \sum_{\substack{ n,q=1\\ (n,q)\ne (t,t)}}^{T}  c_{k} [n] c_{k}^{\ast} [q] f(n,q),
	\end{align}
	where $f(n,q) \triangleq  { \E  \left \{{ \mf  r[n]^{\HH}  \mf  r[t]  \mf  r[t]^{\HH}  \mf  r[q] }\right \}}/M^2$. It should be noted that when $n=q=t$,  $f(t,t) =1$, however, we don't replace $f(t,t)$ by 1 in~\eqref{eq:hat s_kt_squared} and keep it in its general form in order to specialize the result for the unquantized system with $\mf r [t]=\mf y[t]$.
	
	Using the definition of quantized signal, we can express $f(n,q)$  in terms of unquantized signal and QN as follows:
	\begin{align}
		\label{eq:zeta_fn} \nonumber
		f(n&,q) = \frac{1}{M^2} \Big (f_1 +f_2 + 2 \Re \{ f_3 \} + f_4 + f_5+ 2 \Re \{f_6\}  \\
		&+\> 2 \Re \{ f_7 \} + 2 \Re \{ f_8 \}+2 \Re \{ f_9 \}+2 \Re \{ f_{10} \} \Big ),
	\end{align}
	where the terms $f_1-f_{10}$ are given by
	\begin{subequations}
		\begin{align}
			f_1 &=\gamma ^2  \E \left \{ \mf {y}^{\HH}[n]  \mf {y} [t]    \mf {y}^{\HH} [t]    \mf {y} [q] \right \}  \label{eq:a_1}\\
			f_2&=   \E \left \{  \mf {z}^{\HH} [n]    \mf {z} [t]    \mf {z}^{\HH} [t]    \mf {z} [q] \right \}   \\
			f_3&= \gamma   \E \left \{ \mf {y}^{\HH} [n]    \mf {y} [t]    \mf {z}^{\HH} [t]    \mf {z} [q] \right \}  \\
			f_4 &=\gamma   \E \left \{  \mf {y}^{\HH} [n]    \mf {z} [t]     \mf {z} [t]  ^{\HH}   \mf {y} [q] \right \}  \\
			f_{5}&=\gamma   \E \left \{  \mf {z}^{\HH}  [n]   \mf {y} [t]    \mf {y}^{\HH} [t]    \mf {z} [q] \right \}   \\
			f_6&=\gamma ^{\frac{3}{2}}  \E \left \{  \mf {y}^{\HH} [n]    \mf {z} [t]     \mf {y}^{\HH} [t]    \mf {y} [q] \right \} \\
			f_7&= \gamma ^{\frac{3}{2}}  \E \left \{  \mf {y}^{\HH} [n]    \mf {y} [t]    \mf {y}^{\HH} [t]    \mf {z} [q]  \right \}   \\
			f_{8}&=\gamma    \E \left \{  \mf {y}^{\HH} [n]    \mf {z} [t]    \mf {y}^{\HH} [t]    \mf {z} [q] \right \}  \\
			f_{9}&={\gamma }^{\frac{1}{2}}  \E \left \{  \mf {y}^{\HH} [n]    \mf {z} [t]    \mf {z}^{\HH} [t]   \mf {z} [q] \right \}   \\
			f_{10}&= {\gamma }^{\frac{1}{2}}   \E \left \{   \mf {z}^{\HH} [n]    \mf {y} [t]    \mf {z}^{\HH} [t]    \mf {z} [q] \right \}.   \label{eq:a_16}
		\end{align}
	\end{subequations}
	
	It is clear that the evaluation of $f_2-f_{10}$ when $(n,q)\ne (t,t)$ is challenging since $\mf  r[n], \mf  r [t] $ and $ \mf  r [q] $ are not independent, in general. Therefore, in the following we shall make use of Assumption 3 in Sec.~\ref{sec: Assumption} alongside the asymptotic analysis to obtain a closed-form yet a good approximation for $\tilde{\sigma}_{\epsilon_k}^2$. 
	
	Thus, making use of Lemmas~\ref{lem: xxx0}-~\ref{lem: xxx3}, we can show that
	\begin{equation}\label{eq:mean_a_2}
		\frac{f_2}{M}\approx \left\{ \begin{array}{l}
			\sigma_z^4 ,\quad  n=q \ne t\\
			0,\quad \text{else}
		\end{array} \right.
	\end{equation}
	\begin{equation} \label{eq:asymp_a3}
		\frac{f_3}{M^2}  \approx \left\{ \begin{array}{l}
			{\alpha \rho \gamma \sigma_z^2 \bar {\mf c}^{\HH}[t] \bar {\mf c}[n]},\quad q=t,n\ne t \\
			0,\quad \text{else}
		\end{array} \right.
	\end{equation}
	\begin{equation} \label{eq:asymp_a56}
		\frac{f_4}{M}=\frac{f_5}{M}   \approx \left\{ \begin{array}{l}
			{\gamma \sigma_z^2} \sigma_y^2,\quad n=q \ne t\\
			0,\quad \text{else}
		\end{array} \right.
	\end{equation}
	\begin{equation}\label{eq:aymp_a7_16}
		{f_6}={f_7}= \cdots ={f_{10}} \approx 0.
	\end{equation}
	Combining~\eqref{eq:a_1},~\eqref{eq:aymp_a7_16}-~\eqref{eq:mean_a_2},~\eqref{eq:zeta_fn} and~\eqref{eq:hat s_kt_squared} gives
	\begin{align}\label{eq:xxxxxxxxxx}\nonumber
		\E \{|\hat {s}_{k} [t] |^2\} &\approx
		\frac{\xi^2}{M^2} f(t,t) +\frac{\xi^2 \gamma^2}{M^2} \mu_1 + {2 \alpha \rho \xi^2 \sigma_z^2 \gamma} (T-K) \\
		&{}+ \frac{\xi^2 \sigma_z^4}{M} (T-1) + \frac{2 \xi^2 \gamma \sigma_z^2 \sigma_y^2 }{M}(T-1)
	\end{align}
	where $\mu_1 = \E \{ \mf c_k^{\TT} \mf Y^{\HH} \mf y[t] \mf y^{\HH}[t] \mf Y \mf c_k^{\ast}\} -M(M+1)\sigma_y^4 \label{sub:E1}$.
	In~\eqref{eq:xxxxxxxxxx}, the second term is due to $f_1$, third to $f_3$, fourth to $f_2$ and fifth to $f_4 \& f_5$.
	\begin{figure*}[!btp]
		\normalsize
		\begin{IEEEeqnarray}{rCl} \nonumber
			\mu_1 &=& \alpha ^2 \rho ^2 M T^2 (K+M) + \bar {\alpha}^2 \rho ^2 K M (K M+K T+M T+1)  + \bar \alpha \alpha  \rho ^2 K M T (K+M)  + \bar \alpha \alpha \rho ^2 M T^2 (K+M)\\ \nonumber
			&& +\> 2 \bar \alpha \rho K M^2 + 2 \bar \alpha \alpha  \rho ^2 M T (K M+1) + 2 \bar \alpha \rho K M T+\alpha  \rho K M T + 2 \alpha  \rho M^2 T + \alpha \rho M T^2+M (M+T) \\
			&&-\>M (M+1) (K \rho +1)^2.
			\label{eq:mu1_value}
		\end{IEEEeqnarray}
	\end{figure*}
	\begin{figure*}[!btp]
		\normalsize
		\begin{IEEEeqnarray}{rCl}  \nonumber
			\tilde{\sigma}_{\epsilon_k}^2 &=&\biggl( M^2 f(t,t) - \gamma ^2 M \Big(\rho ^2 \left(K^2 (M-\bar \alpha T+1)-K \left(\bar \alpha M T+\alpha  \left(\alpha +T^2-2\right)+1\right)-2 \bar \alpha \alpha  T\right)+\rho  (2 K (M-T+1)\\
			&&+\>\alpha  T (K-T))+M-T+1\Big)+ M (T-M-1) \left(2 \gamma( K \rho+1 )+\sigma _z^2\right)\sigma _z^2 \biggl) \Big / \Big({\alpha  \bar{\alpha}  \rho ^2 \gamma^{ 2} M^2 T^2}\Big)
			\label{eq:noise_var_QSP}
		\end{IEEEeqnarray}
		\hrulefill
	\end{figure*}
	After some mathematical manipulations, we can verify that $\mu_1$ is given by~\eqref{eq:mu1_value} which is shown at the top of page~\pageref{eq:mu1_value}.
	
	Substituting~\eqref{eq:mu1_value} in~\eqref{eq:xxxxxxxxxx} and combining the result with~\eqref{eq:Expected_shat},~\eqref{eq:def_ak} and~\eqref{eq:normalized_var_eff_noise} yields the closed-form expression for $\tilde{\sigma}_{\epsilon_k}^2$ given in~\eqref{eq:noise_var_QSP}, which is shown on the upper half of page~\pageref{eq:noise_var_QSP}. Finally, substituting~\eqref{eq:noise_var_QSP} with $f(t,t) = 1$ in~\eqref{eq:Stnd_lowerbound} yields~\eqref{eq: MRC_Rate_QSP}. This completes the proof.
	
	\section{Proof of Theorem 2}
	\label{app: proof_thm2}
	In this section, we prove Theorem~\ref{thm: MRC_Rate_QSP_multiplecell}. We begin by using the decomposition in~\eqref{eq:MRC_decomposition}. Redefining $a_k, s_k[t]$ and $\epsilon_{k}[t]$ as $a_{0 k}, s_{0 k}[t]$ and $\epsilon_{0 k}[t]$, respectively,~\eqref{eq:multicell_MRC_out} can thus be written as $ \hat{s}_{0 k}[t]= a_{0 k} s_{0 k}[t]+ \epsilon_{0 k}[t]$. Let
	\begin{equation}
		\tilde{\sigma}_{\epsilon_{0 k}}^2 = \frac {{\E \{|\hat {s}_{0 k} [t] |^2\}- |\E\{\hat{s}_{0 k}[t] \}|^2}}{| a_{0 k}  |^2}- 1
		\label{eq:norm_noise_multicell}
	\end{equation}
	
	be the \textit{normalized} variance of effective noise at the output of MRC and hence the lower bound on achievable rate is
	\begin{align}  \nonumber
		\label{eq:avg_rate_multicell}
		R_{\text{multicell}}^{\text{LB}} &= \E \left\{ \log \left( 1 + {1}/{ \tilde{\sigma}_{\epsilon_{0 k}}^2} \right) \right \} \\
		&\ge \log \left( 1 + {1}/{\E \{\tilde{\sigma}_{\epsilon_{0 k}}^2\}} \right),
	\end{align}
	where the expectation is taken w.r.t. the large-scale fading and the second line follows from applying the Jensen's inequality on the convex function $\log \left(1 + {1}/{ \tilde{\sigma}_{\epsilon_{0 k}}^2} \right)$ w.r.t. the random variable $\tilde{\sigma}_{\epsilon_{0 k}}^2$.
	
	Thus our task is to calculate~\eqref{eq:norm_noise_multicell}.  Following the same lines of proof in Appendix~\ref{app: proof_thm1}, it is easy to show that $a_{0 k}$ and $\E\{\hat{s}_{0 k}[t] \}$ are given by:
	\begin{subequations}
		\begin{align}
			a_{0 k}  &\approx  \xi^{\prime} \gamma^{\prime} \rho T \sqrt{\bar \alpha \alpha } \label{eq:a0k}\\
			\E\{\hat{s}_{0 k}[t] \}&=\sqrt{\alpha \rho   \gamma^{\prime} } c_{0 k}[t],
		\end{align}
		\label{eq:a0k_expected}
	\end{subequations}
	where the approximation in~\eqref{eq:a0k} is due to Lemma~\ref{lem: xxx1}.
	
	Using~\eqref{eq:hat s_kt_squared} and after redefining all parameters according to the signal model~\eqref{eq:QBB_multiplecell}, $\E \{|\hat {s}_{0 k} [t] |^2\}$ can be written as
	\begin{align} \nonumber
		\label{eq:hat s_kt_squared_multicell}
		\E\{|\hat {s}_{0 k} [t] |^2\} &= {\xi^{\prime 2}} f^{\prime}(t,t)  \\
		&+ \xi^{ \prime 2} \sum_{\substack{ n,q=1 \\ (n,q)\ne (t,t)}}^{T}  c_{0 k} [n] c_{0 k}^{\ast} [q] f^{\prime}(n,q),
	\end{align}
	where $f^{\prime}(n,q) \triangleq { \E  \left \{{ \mf  r_0[{n}]^{\HH}  \mf  r_0[t]  \mf  r_0[t]^{\HH}  \mf  r_0[q] }\right \}}/M^2$ and $\xi^{ \prime }$ is given in~\eqref{eq:xi_mutiplecell}. Expanding $f^{\prime}(n,q)$ as we have done previously in~\eqref{eq:zeta_fn}, we get
	\begin{IEEEeqnarray}{rCl} \nonumber
		f^{\prime}(n&,&q)= \frac{1}{M^2} \Big (f^{\prime}_1 +f^{\prime}_2 + 2 \Re \{ f^{\prime}_3 \} + f^{\prime}_4 + f^{\prime}_5+ 2 \Re \{f^{\prime}_6\}\\
		&&+\> 2 \Re \{f^{\prime}_7 \} + 2 \Re \{ f^{\prime}_8 \}+2 \Re \{ f^{\prime}_9 \}+2 \Re \{ f^{\prime}_{10} \} \Big ),
		\label{eq:zeta_fn_multicell}
	\end{IEEEeqnarray}
	where we redefine $\{f_j\}$ as $\{f^{\prime}_j\}$ with replacing all quantized signal and QN vectors in~\eqref{eq:a_1}-~\eqref{eq:a_16} according to~\eqref{eq:QBB_multiplecell}. Following the same lines of proof as previously,~\eqref{eq:hat s_kt_squared_multicell} can be readily written as
	\begin{IEEEeqnarray}{rCl} \nonumber
		\E \{|\hat {s}_{0 k} && [t] |^2\}\approx
		{\xi^{\prime 2}} f^{\prime}(t,t) + \frac{\xi^{\prime 2} \gamma^{\prime 2} }{M^2} \mu^{\prime}_1 + {2 \alpha \rho {\xi^{\prime 2}} \sigma_z^2 \gamma^{\prime} } (T-\kappa_0)\\
		&&{}+ \frac{{\xi^{\prime 2}} \sigma_z^4}{M} (T-1) + \frac{2 {\xi^{\prime 2}}  \gamma^{\prime}  \sigma_z^2 (\kappa_0 \rho +1)^2 }{M}(T-1),
		\label{eq:xxxxxxxxxx_multicell}
	\end{IEEEeqnarray}
	where $\mu^{\prime}_1$ is given by
	\begin{IEEEeqnarray}{rCl} \nonumber
		\mu^{\prime}_1 &=& \E \{ \mf c_{0 k}^{\TT} \mf Y_0^{\HH} \mf y_0[t] \mf y_0^{\HH}[t] \mf Y_0 \mf c_{0 k}^{\ast}\} -M(M+1)\sigma_{y_0}^4 \\ \nonumber
		&=&M \Big (  \Big(2 \alpha  M \rho ^2 T-2 \alpha  M \rho -2 \alpha ^2 M \rho ^2 T-2 \rho +\alpha  \rho ^2 T^2\\ \nonumber
		&&-\alpha  \rho  T+2 \rho  T\Big)\kappa _0  + \Big(\alpha ^2 M \rho ^2-2 \alpha  M \rho ^2-\rho ^2-\alpha  \rho ^2 T\\ \nonumber
		&&+\rho ^2 T\Big)\kappa _0^2 +\bar \alpha \rho ^2 (MT+\bar \alpha)\kappa _1 +\alpha  M \rho ^2 T^2+2 \alpha  M \rho  T\\
		&&+\alpha  \rho  T^2-2 \alpha ^2 \rho ^2 T+2 \alpha  \rho ^2 T+T-1
		\Big).
		\label{sub:E1_multiplecell}
	\end{IEEEeqnarray}
	
	Combining~\eqref{sub:E1_multiplecell},~\eqref{eq:xxxxxxxxxx_multicell},~\eqref{eq:a0k_expected} and~\eqref{eq:norm_noise_multicell},  and after some mathematical manipulations and plugging all related equations in~\eqref{eq:norm_noise_multicell}, we obtain~\eqref{eq:normalized_noise_multiplecell}, which is shown on the upper half of page~\pageref{eq:normalized_noise_multiplecell}. Substituting~\eqref{eq:gamma_muliplecell},~\eqref{eq:xi_mutiplecell}, and $f^{\prime} (t,t) =1$ in~\eqref{eq:normalized_noise_multiplecell}, we obtain the final expression of $\tilde{\sigma}_{\epsilon_{0 k}}^2$ under the 1-bit quantization per single realization of large-scale fading. Finally, the lower bound~\eqref{eq: MRC_Rate_QSP_multiplecell} follows from~\eqref{eq:avg_rate_multicell}. This completes the proof.
	\begin{figure*}[!btp]
		\normalsize
		\begin{IEEEeqnarray}{rCl} \nonumber
			\tilde{\sigma}_{\epsilon_{0 k}}^2   &&=  \biggl(M^2 f^{\prime} (t,t)+ \gamma^{\prime 2} M \rho  \left(\kappa _0 \left(\kappa _0 \rho  (\bar{\alpha}   T -M-1)-2 M+T (\alpha +2 \bar{\alpha}  +\alpha  \rho  T)-2\right)+\bar{\alpha}   \rho  (\bar{\alpha}  +M T)\kappa _1 \right) \\
			&&+\> \gamma^{\prime 2} M (-M+T (\alpha  \rho  (2 \bar{\alpha}   \rho +T)+1)-1)+ M(T-M - 1) \left(2 \gamma^{\prime}  (\kappa _0 \rho +1) +\sigma _z^2\right) \sigma _z^2 \biggl) \Big / \Big({\alpha  \bar{\alpha}   \gamma^{\prime 2} M^2 \rho ^2 T^2}\Big)
			\label{eq:normalized_noise_multiplecell}
		\end{IEEEeqnarray}
		\hrulefill
		\vspace*{0pt}
	\end{figure*}
	
\end{appendices}

%
\bibliographystyle{IEEEtran}
\bibliography{IEEEabrv,QSP_Ref}

\end{document}